\documentclass[12pt]{article}
\usepackage{amsmath}
\usepackage{graphicx}
\usepackage{natbib}
\usepackage{url} 

\newcommand{\blind}{0}

\usepackage{color}
\usepackage{babel}
\usepackage{float}
\usepackage{amsthm}
\usepackage{amssymb}
\usepackage{setspace}
\usepackage[unicode=true,pdfusetitle,
 bookmarks=true,bookmarksnumbered=false,bookmarksopen=false,
 breaklinks=false,pdfborder={0 0 1},backref=page,colorlinks=true]
 {hyperref}
\hypersetup{citecolor=Emerald, linkcolor=RoyalBlue}

\makeatletter

\floatstyle{ruled}
\newfloat{algorithm}{tbp}{loa}
\providecommand{\algorithmname}{Algorithm}
\floatname{algorithm}{\protect\algorithmname}

\numberwithin{equation}{section}
\numberwithin{figure}{section}
  \theoremstyle{plain}
  
  \theoremstyle{plain}
  \newtheorem{ass}{\protect\assname}
  \newtheorem{prop}{\protect\propositionname}
  \theoremstyle{plain}
  \newtheorem{rem}{\protect\remarkname}

\usepackage[dvipsnames]{xcolor}
\AtBeginDocument{%
\renewcommand\equationautorefname{\@gobble}
}

\makeatother

  \providecommand{\lemmaname}{Lemma}
  \providecommand{\propositionname}{Proposition}
  \providecommand{\remarkname}{Remark}
  \providecommand{\assname}{Assumption}

\begin{document}

\def\spacingset#1{\renewcommand{\baselinestretch}%
{#1}\small\normalsize} \spacingset{1}


\if0\blind
{
  \title{\bf Online Smoothing for Diffusion Processes
  Observed with Noise
}
  \author{Shouto Yonekura\thanks{
    SY acknowledges funding from the Alan Turing Institute under grant number TU/C/000013.
}\hspace{.2cm}\\
    Department of Statistical Science, University College London, \\
    Alan Turing Institute,\\
    and \\
    Alexandros Beskos\thanks{AB acknowledges funding via a Leverhulme Trust Prize.} \\
    Department of Statistical Science, University College London,\\
    Alan Turing Institute,}
  \maketitle
} \fi

\if1\blind
{
  \bigskip
  \bigskip
  \bigskip
  \begin{center}
    {\LARGE\bf Online Smoothing for Discretely Observed Jump Diffusions
}
\end{center}
  \medskip
} \fi

\bigskip
\begin{abstract}
We introduce a methodology for online estimation of smoothing expectations 
for a class of additive functionals, in the context of a rich family of diffusion processes 
(that may include jumps) -- observed at discrete-time instances. We overcome the unavailability of the transition density of the underlying SDE by working on the augmented pathspace. 
The new method can be applied, for instance, to carry out online parameter inference for the designated class of models. 
Algorithms defined on the infinite-dimensional pathspace have been developed the last years mainly in the context of MCMC techniques. There, the main  benefit is the achievement of mesh-free mixing times for the practical time-discretised algorithm used on a PC.  Our own methodology sets up the framework for infinite-dimensional online filtering --  an important positive practical consequence is the construct of estimates with variance that does not increase with decreasing mesh-size.  
Besides regularity conditions, our method is, in principle, applicable under the weak assumption  -- relatively to restrictive conditions often required in the MCMC or filtering literature of methods defined on  pathspace -- that the SDE covariance matrix is invertible.
\end{abstract}

\noindent%
{\it Keywords:}  Jump Diffusion; Data Augmentation; Forward-Only Smoothing; Sequential Monte Carlo; Online Parameter Estimation.
\vfill






\newpage
\spacingset{1.5} 

\section{Introduction\label{sec:Intro}}

Research in Hidden Markov Models (HMMs) has -- thus far -- provided effective online algorithms for the estimation of expectations of the smoothing distribution for the case of a class of additive functionals of the underlying signal. Such methods necessitate knowledge of the transition density of the Markovian part of the model between observation times. We carry out a related exploration for the (common in applications) case when the signal corresponds to a diffusion process, thus we are faced with the challenge that such transition densities are typically unavailable. Standard data augmentation schemes that work with the multivariate density of a large enough number of imputed points of the continuous-time signal will lead to ineffective algorithms. The latter will have the abnormal characteristic that -- for given Monte-Carlo iterates -- the variability of the produced estimates will increase rapidly as the resolution of the imputation becomes finer. One of the ideas underpinning the work in this paper is that development of effective algorithms instead requires respecting the structural properties of the diffusion process, thus we build up imputation schemes on the infinite-dimensional diffusion pathspace itself. As a consequence, the time-discretised algorithm used in practice on a PC will be stable under mesh-refinement.    

We consider continuous-time jump-diffusion models observed at discrete-time instances. 
The $d_x$-dimensional process, $X=\{X_{t};t\ge 0\}$, $d_x\ge 1$, is defined via
the following time-homogeneous stochastic differential
equation (SDE), with $X_{t-}:=\lim_{s\uparrow t}X_{t}$,
\begin{align}
dX_{t} =b(X_{t-})dt+\sigma(X_{t-})dW_{t}+dJ_{t},\quad X_{0}=x_0\in\mathbb{R}^{d_x},\,\,\,t\ge 0. \label{eq:SDE}
\end{align}
Solution $X$ is driven by the 
$d_w$-dimensional Brownian motion, $W=\{W_{t};t\ge 0\}$, $d_w\ge 1$, and the compound Poisson process,
$J=\{J_t;t\ge 0\}$. The SDE involves a drift function $b=b_{\theta}:\mathbb{R}^{d_x}\mapsto\mathbb{R}^{d_x}$ and coefficient matrix 
$\sigma=\sigma_{\theta}:\mathbb{R}^{d_x}\mapsto\mathbb{R}^{d_x\times d_w}$, for parameter $\theta\in\mathbb{R}^{p}$, $p\ge 1$.
Let $\{N_{t};t\ge 0\}$
be a Poisson process with intensity function $\lambda_{\theta}(\cdot)$, and 
 $\{\xi_{k}\}_{k\geq 1}$ i.i.d.~sequence of  
 random variables with Lebesgue density $h_{\theta}(\cdot)$; the c\`adl\`ag process $J$ is determined as $J_{t}= J_{t,\theta}=\sum_{i=1}^{N_{t}}\xi_{i}$.
We work under standard assumptions (e.g.\@ linear growth, Lipschitz
continuity for $b$, $\sigma$) that guarantee a unique global solution of (\ref{eq:SDE}), in a weak or strong sense,
see e.g.~\citet{oksendal2007applied}.  

SDE (\ref{eq:SDE}) is  observed with noise at 
 discrete-time instances
$0 = t_0 < t_{1}<t_{2}<\cdots <t_{n}$, $n\ge 1$. 
Without loss of generality, we assume equidistant observation times, 
with $\Delta:= t_1-t_{0}$.
We consider data $Y_{t_0},\ldots, Y_{t_n}$, and for simplicity we set,
\begin{gather*}
x_i:=X_{t_i},\quad  y_i=Y_{t_i}, \quad  0\le i\le n. 
\end{gather*}
Let $\mathcal{F}_{i}$ be a filtration generated by $X_{s}$ for $s\in[t_{i-1},t_i], 0< i\le n$. We assume,
\begin{align}
\big[\,Y_{t_i} \,\big|\,\{Y_{t_{j}};j< i\},\,\{X_{s};s\in[0,t_i]\}\,\big] \sim
g_{\theta}\big(dY_{t_i}\,\big|\,Y_{t_{i-1}},\mathcal{F}_{i}\big),\quad  0\le i \le n,  \label{eq:likelihood}
\end{align}
for conditional distribution $g_{\theta}(\cdot|Y_{t_{i-1}},\mathcal{F}_i)$ on $\mathbb{R}^{d_y}$, $d_y\ge 1$, under convention $Y_{t_{-1}}=y_{-1}=\emptyset$.
%
We write, 
\begin{align}
[\,x_{i} \,|\, x_{i-1}\,] \sim f_{\theta}(dx_{i}|x_{i-1}),
\label{eq:trans}
\end{align}
where $f_{\theta}(dx_i|x_{i-1})$ is the transition distribution of the driving SDE process (\ref{eq:SDE}).
We consider the density functions of $g_{\theta}(dy_i|y_{i-1},\mathcal{F}_i)$ 
and $f_{\theta}(dx_{i}|x_{i-1})$, and -- with some abuse of notation -- we  
write $g_{\theta}(dy_{i}|y_{i-1},\mathcal{F}_i)= g_{\theta}(y_{i}|y_{i-1},\mathcal{F}_i)dy_i$, $f_{\theta}(dx_{i}|x_{i-1})= f_{\theta}(x_i|x_{i-1})dx_i$, where $dy_i$, $dx_i$ denote Lebesgue measures. Our work develops under the following regime.
\begin{ass}
\label{ass:unknown}
The transition density $f_{\theta}(x'|x)$ is intractable;
the density $g_{\theta}(y'|y,\mathcal{F})$ is analytically available --
for appropriate $x'$, $x$, $y'$, $y$, $\mathcal{F}$ consistent with  preceding definitions.
\end{ass}
\noindent The intractability of the transition density $f_{\theta}(\cdot|\cdot)$ will pose challenges for the main  problems this paper aims to address.

Models defined via (\ref{eq:SDE})-(\ref{eq:trans}) are extensively used, e.g.,  in finance and
econometrics, for instance for capturing the market microstructure noise,
see \citet{ait2005often,ait2008high,hansen2006realized}.
The above setting belongs to the general class of HMMs, with a signal defined in continuous-time. See \citet{cappe2005inference,douc2014nonlinear}
for a general treatment of HMMs fully specified in discrete-time. A number of methods have been suggested in the literature 
for approximating the unavailable transition density -- mainly in the case of processes without jumps -- including:
asymptotic expansion techniques \citep{ait2005often,ait2002maximum,ait2008closed,kessler1997estimation};
martingale estimating functions \citep{kessler1999estimating};
generalized method of moments \citep{hansen1993back}; Monte-Carlo approaches \citep{wagn:89,durh:02, beskos2006exact}.
See, e.g., \citet{kessler2012statistical} for a detailed review.

For a given sequence $\{a_{m}\}_{m}$, we use the notation $a_{i:j}:=(a_{i},\ldots,a_{j})$, 
for integers $i\leq j$. 
Let $p_{\theta}(y_{0:n})$ denote the joint density of $y_{0:n}$. 
Throughout the paper, $p_{\theta}(\cdot)$ is used generically to represent probability distributions or densities of random variables appearing 
as its arguments.

Consider the maximum likelihood estimator (MLE),
\begin{align*}
\hat{\theta}_{n} & :=\arg\max_{\theta\in\Theta}\log p_{\theta}(y_{0:n}).
\end{align*}
Except for limited cases, one cannot obtain the MLE analytically
for HMMs (even for discrete-time signal) due to the intractability of 
$p_{\theta}(y_{0:n})$.

We have set up the modelling context for this work. 
The main contributions of the paper 
in this setting -- several of which relate with overcoming the intractability of the transition density of the SDE, and developing a methodology that is well-posed on the infinite-dimensional pathspace -- will be as follows:
\begin{itemize}
\item[(i)] We present an \emph{online} algorithm that delivers Monte-Carlo estimators of smoothing expectations,
\begin{equation}
\label{eq:exp}
\mathcal{S}_{\theta,n}=\int S_{\theta}(\mathbf{x}_{0:n}) p_{\theta}(d\mathbf{x}_{0:n}|y_{0:n}), \quad n\ge 1,
\end{equation}
 for the class of additive functionals $S_{\theta}(\cdot)$ of the structure, 
\begin{equation}
\label{eq:sum} 
 S_{\theta}(\mathbf{x}_{0:n})=\sum_{k=0}^{n}s_{\theta,k}(x_{k-1},\mathbf{x}_{k}),
\end{equation} 
under the conventions $x_{-1}=\emptyset$, $\mathbf{x}_0 = x_0$.
The bold type notation $\mathbf{x}_k$, $k\ge 0$, is reserved for carefully defined variables involving elements of the infinite-dimensional pathspace. The specific construction of $\mathbf{x}_k$ will depend on the model at hand, and will be explained in the main part of the paper. We note that the class of additive functionals considered supposes additivity over time, and therefore includes the general case of evaluating the score function.
The online solution of the smoothing problem is often used as the means to solve some concrete inferential problems - see, e.g., (ii) below.

\item[(ii)] We take advantage of the new approach to show numerical applications, with emphasis on
carrying out \emph{online} parameter inference for the designated class of models via a gradient-ascent  approach (in a Robbins-Monro stochastic gradient framework). A critical aspect of this particular online algorithm (partly likelihood based, when concerned with parameter estimation; partly Bayesian, with regards to identification of filtering/smoothing expectations) is that it delivers estimates of the evolving score function,  of the model parameters, together with particle representations of the filtering distributions, through a \emph{single passage} of the data. 
This is a unique favourable algorithmic characteristic, when constrasted with alternative algorithms with similar objectives, such as, e.g.,  Particle MCMC \citep{andrieu2010particle}, or SMC$^2$ \citep{chop:13}.   

\item[(iii)] In this work, we will \emph{not} characterise analytically the size of the
time-discretisation bias relevant to the SDE models at hand, and are content that: (I) the bias can be decreased by increasing the resolution of 
the numerical scheme (typically an Eyler-Maruyama one,
or some other Taylor scheme, see e.g.~\cite{kloe:13});
(II) critically, the Monte-Carlo algorithms are developed in a manner that the variance does not increase (in the limit, up to infinity) when increasing the resolution of the time-discretisation method; to achieve such an effect, the  
algorithms are (purposely) defined on the infinite-dimensional pathspace, and SDE paths are only discretised when implementing the algorithm on a PC (to allow, necessarily, for finite computations).

\item[(iv)] Our method draws inspiration from earlier works, in the context of online filtering for discrete-time HMMs and infinite-dimensional pathspace MCMC methods. The complete construct is novel; one consequence of this is that it is applicable, in principle, for a wide class of SDEs, under the following weak assumption (relatively to restrictive conditions often imposed in the literature of infinite-dimensional MCMC methods).
\begin{ass}
\label{ass:invert}
The diffusion covariance matrix function,
\begin{equation*}
\Sigma_{\theta}(v):=\sigma_{\theta}(v)\sigma_{\theta}(v)^{\top}\in\mathbb{R}^{d_x\times d_x},
\end{equation*}
 is invertible, for all relevant $v$, $\theta$.
\end{ass}
\end{itemize}
Thus, the methodology does not apply as defined here only for the class of hypoelliptic SDEs. 

An elegant solution to the online smoothing problem posed above in (i), for the case of a standard HMM with discrete-time 
signal of known transition density $f_{\theta}(x'|x)$, is given in \cite{del2010forward, poyi:11}. Our own work overcomes 
 the unavailability of the transition density in the continuous-time scenario by following the above literature but augmenting the hidden state with the complete continuous-time SDE path. Related augmentation approaches in this setting -- though for different inferential objectives -- have appeared in 
 \cite{fearnhead2008particle, stro:09, gloaguen2018online}, where the auxiliary variables are derived via the Poisson estimator of transition densities for SDEs  (under strict conditions on the class of SDEs; no jumps), introduced in \cite{beskos2006exact},
 and in \cite{sark:08} where the augmentation involves indeed the continuous-time path (the objective therein is to solve the filtering problem and the method is applicable for SDEs with additive Wiener noise; no jump processes are considered).   

\textbf{A Motivating Example:} 
Fig.~\ref{fig:non-stable} (Top-Panel) shows
estimates of the score function, evaluated at the true parameter value $\theta=\theta^{\dagger}$, for parameter $\theta_3$ 
of the
Ornstein--Uhlenbeck (O--U) process, 
$dX_t = \theta_1(\theta_2-X_t)dt + \theta_3 dW_t$, $X_0=0.0$, for $n=10$ observations $y_i = x_i + \epsilon_i$, $\epsilon_i\stackrel{i.i.d}{\sim}\mathcal{N}(0,0.1^2)$, $1\le i \le n$. Data were simulated from 
$\theta^{\dagger}=(0.5,0.0,0.4)$ with an Euler-Maruyama scheme of $M^{\dagger}=1,000$ grid points per unit of time.
Fig.~\ref{fig:non-stable} (Top-Panel) illustrates the `abnormal' effect of a standard data-augmentation scheme, where for $N=100$ particles, the Monte-Carlo method (see later sections for details) produces estimates of increasing variability as algorithmic resolution increases with 
$M=10, 50, 100, 200$ -- i.e., as it approaches the `true' resolution used for the data generation. In contrast, Fig.~\ref{fig:non-stable} (Bottom Panel) shows results for the same score function as estimated by our proposed method in this paper. 
Our method is well-defined on the infinite-dimensional pathspace, and, consequently, it manages to robustly estimate  the score function under mesh-refinement.

\begin{figure}[!h]
\begin{centering}
\includegraphics[scale=0.5]{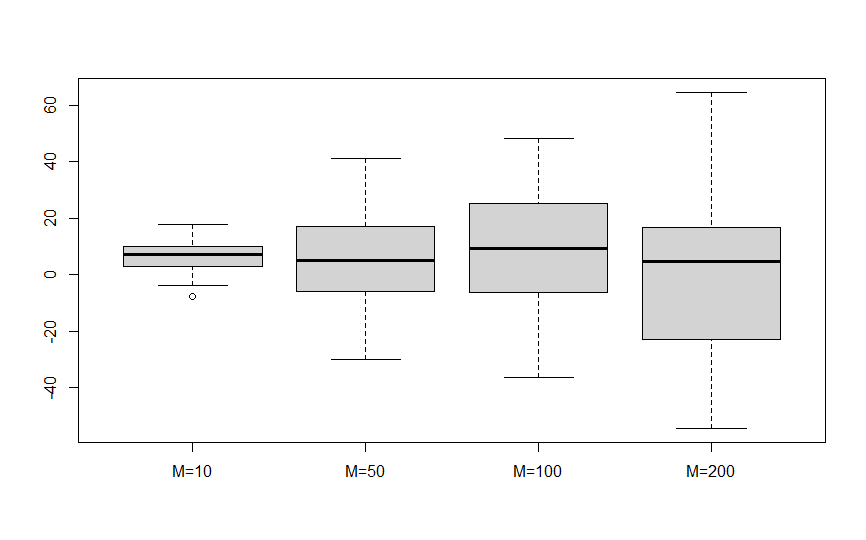}\\[-0.7cm]
\includegraphics[scale=0.5]{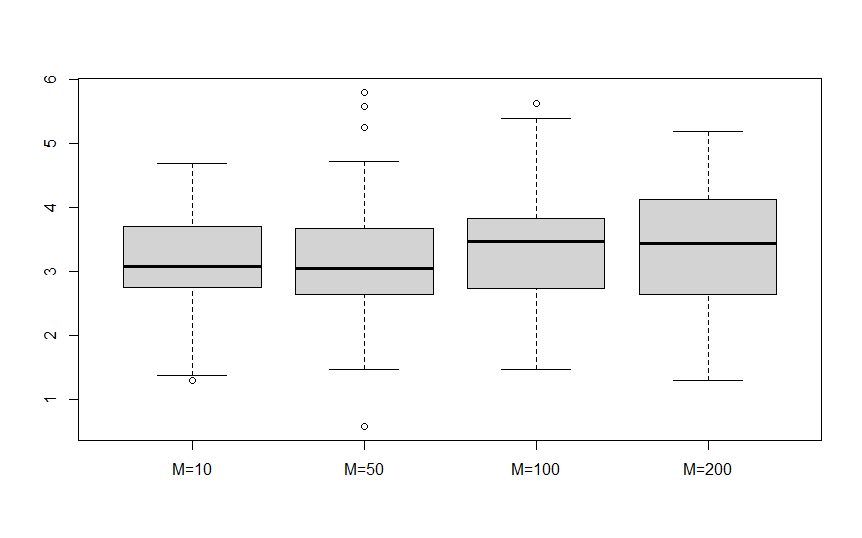}
\par\end{centering}
\vspace{-0.5cm}
\caption{Boxplots of estimated score functions of $\theta_{3}$ for the O--U process, for different resolution sizes $M$, 
over $R=50$ experiment replications. 
$N=100$ particles 
were used in all cases, for the same $n=10$ data-points.
\emph{Top-Panel}: Results are produced  
via a Monte-Carlo algorithm that uses a naive `finite-dimensional' data augmentation approach.
\emph{Bottom-Panel}: Results are produced  
via a Monte-Carlo algorithm that uses the `infinite-dimensional' data augmentation approach introduced in this paper.}
\label{fig:non-stable}
\end{figure}

The rest of the paper is organised as follows.
Section \ref{sec:online-smoothing} reviews the Forward-Only algorithm for the online approximation of expectations of a class of additive functionals described in \cite{del2010forward}.
Section~\ref{sec:SDE and Girsanov} sets up the framework for the treatment of pathspace-valued SDEs, first for the conceptually simpler case of SDEs without jumps, and then proceeding with incorporating jumps.
Section~\ref{sec:cont} provides the complete online approximation algorithm, constructed on the infinite-dimensional pathspace.
Section \ref{sec:online} discusses the adaptation of the developed methodology for the purposes of online inference
for unknown parameters of the given SDE model. 
Section \ref{sec:numerics} shows numerical applications of the developed methodology.
Section \ref{sec:Conclusion-and-remarks} contains conclusions and directions for future research. 
\section{Forward-Only Smoothing}
\label{sec:online-smoothing}

A bootstrap filter \citep{gordon1993novel} is applicable
in the continuous-time setting, as it only requires forward sampling of the 
underlying signal $\{X_t\}_{t\ge 0}$; this is trivially possible -- under numerous approaches -- and is typically associated with the introduction of some time-discretisation bias \citep{kloe:13}. However, the transition density is still required for the smoothing problem we have posed in the Introduction. In this section, we assume 
a standard discrete-time HMM, with initial distribution $p_{\theta}(dx_0)$, transition density $f_{\theta}(x'|x)$, 
and likelihood $g_{\theta}(y|x)$, for appropriate $x, x'\in \mathbb{R}^{d_x}$, 
$y\in \mathbb{R}^{d_y}$,
and review the online algorithm developed in 
\cite{del2010forward} for this setting. Implementation of the bootstrap filter provides 
an approximation of the smoothing distribution $p_{\theta}(dx_{0:n}|y_{0:n})$
by following the geneology of the particles. This method is studied, e.g., in \citet{cappe2009online,dahlhaus2010particle}.
Let $\{x_{0:n}^{(i)},W_{n}^{(i)}\}_{i=1}^{N}$, $N\ge 1$,
be a particle approximation of the smoothing distribution $p_{\theta}(dx_{0:n}|y_{0:n})$, in
the sense that we have the estimate,
\begin{align}
\widehat{p}_{\theta}(dx_{0:n}| y_{0:n})=\sum_{i=1}^{N}W_{n}^{(i)}\delta_{x_{0:n}^{(i)}}(dx_{0:n}),\quad \sum_{i=1}^{N}W_{n}^{(i)}=1,
\label{eq:path-approximation}
\end{align}
%
with $\delta_{x_{0:n}^{(i)}}(dx_{0:n})$ 
 the Dirac
measure with an atom at $x_{0:n}^{(i)}$. 
Then, replacing  $p_{\theta}(x_{0:n}|y_{0:n})$ with 
its estimate in (\ref{eq:path-approximation}) provides consistent estimators of expectations of the HMM smoothing distributions. 
Though the
method is online and the computational cost per time step is $\mathcal{O}(N)$, it typically suffers from
the well-documented \emph{path-degeneracy} problem --
as illustrated via theoretical results or numerically \citep{del2010forward, kantas2015particle}. 
That is, as $n$ increases, the particles 
representing $p_{\theta}(dx_{0:n}| y_{0:n})$ obtained by the
above method will eventually all share the same ancestral particle due
to the resampling steps, and the approximation
collapses for big enough~$n$. This is well-understood not to be a solution to the approximation of the smoothing distribution for practical applications. 

An approach which overcomes path-degeneracy is the  
Forward Filtering Backward Smoothing (FFBS) algorithm of \cite{doucet2000sequential}. 
We briefly review the method here, following closely the notation and development in 
\cite{del2010forward}. In the forward direction, assume that a filtering algorithm (e.g.~bootstrap) has 
provided a particle approximation of the filtering distribution $p_{\theta}(dx_{k-1}|y_{0:k-1})$ -- assuming a relevant $k$ --,
\begin{align}
\label{eq:p-filter}
\widehat{p}_{\theta}(dx_{k-1}|y_{0:k-1})= \sum_{i=1}^{N}W_{k-1}^{(i)}\delta_{x_{k-1}^{(i)}}(dx_{k-1}),
\end{align}
for weighted particles $\{x_{k-1}^{(i)},W_{k-1}^{(i)}\}_{i=1}^{N}$.
In the backward direction, assume that one is given the particle approximation 
of the marginal smoothing distribution $p_{\theta}(dx_k|y_{0:n})$, 
\begin{align}
\widehat{p}_{\theta}(dx_{k}|y_{0:n}) = \sum_{i=1}^{N}W_{k|n}^{(i)}\delta_{x_k^{(i)}}(dx_k).
\label{eq:p-back}
\end{align}
One has that \citep{kitagawa1987non},
\begin{align}
p_{\theta}(dx_{k-1:k}| y_{0:n}) &= p_{\theta}(dx_k|y_{0:n})\otimes p_{\theta}(dx_{k-1}|x_{k},y_{0:k-1}) \nonumber \\  
&= p_{\theta}(dx_k|y_{0:n})\otimes p_{\theta}(dx_{k-1}| y_{0:k-1})\frac{f_{\theta}(x_{k}| x_{k-1})}{\int f_{\theta}(x_{k}| x_{k-1})p_{\theta}(x_{k-1}|y_{0:k-1})dx_{k-1}}.
\label{eq:Kitagawa-Recursion}
\end{align}
Using (\ref{eq:p-filter})-(\ref{eq:p-back}), and based on equation (\ref{eq:Kitagawa-Recursion}),
we obtain the 
approximation, 
\begin{align}
\widehat{p}_{\theta}(dx_{k-1:k}| y_{0:n}) = \sum_{j=1}^{N}W_{k|n}^{(j)}\sum_{i=1}^{N}\frac{f_{\theta}(x_{k}^{(j)}|x_{k-1}^{(i)})W_{k-1}^{(i)}}{\sum_{l=1}^{N}f_{\theta}(x_{k}^{(j)}| x_{k-1}^{(l)})W_{k-1}^{(l)}}\delta_{(x_{k-1}^{(i)},x_k^{(j)})}(dx_{k-1:k}).\label{eq:SMC-Smoothing}
\end{align}
Recalling the expectation of additive functionals in (\ref{eq:exp})-(\ref{eq:sum}) -- where now, in the discrete-time setting, we can ignore the bold $\mathbf{x}_{k}$ elements, and 
simply use $x_k$ instead -- 
the above calculations give rise to the following estimator of the target quantity $\mathcal{S}_{\theta,n}$ in (\ref{eq:sum}),
\begin{equation*}
\widehat{\mathcal{S}}_{\theta,n} = \sum_{k=0}^{n} \int s_{\theta,k}(x_{k-1},x_k)\,\,\widehat{p}_{\theta}(dx_{k-1:k}| y_{0:n}). 
\end{equation*}
To be able to apply the above method, the marginal smoothing approximation in (\ref{eq:p-back}) is 
obtained via a backward recursive approach. In particular, 
starting from $k=n$ (where the approximation is provided by the standard forward particle filter), one proceeds as follows. 
Given $k$, the quantity for $k-1$ is directly obtained 
by integrating out $x_k$ in (\ref{eq:SMC-Smoothing}), 
thus we have, 
\begin{align*}
\widehat{p}_{\theta}(dx_{k-1}|y_{0:n}) = \sum_{i=1}^{N} W_{k-1|n}^{(i)} \delta_{x_{k-1}^{(i)}}(dx_{k-1}), 
\end{align*}
for the normalised weights, 
\begin{align*}
W_{k-1|n}^{(i)} \propto \sum_{j=1}^{N}W_{k|n}^{(j)}
\,\frac{f_{\theta}(x_k^{(j)}|x_{k-1}^{(i)})W_{k-1}^{(i)}}{\sum_{l=1}^{N}f_{\theta}(x_{k}^{(j)}| x_{k-1}^{(l)})W_{k-1}^{(l)}}. 
\end{align*}
%
Notice that -- in this version of FFBS -- 
the same particles $\{x_{k}^{(i)}\}_{i=1}^{N}$ are used in both directions (the ones before resampling at the forward filter), but with different weights.

An important development made in \cite{del2010forward} is transforming the  above \emph{offline} algorithm 
 into an \emph{online} one. This is achieved by consideration of the sequence of  instrumental functionals,   
\begin{align*}
T_{\theta,0}(x_0)= s_{\theta,0}(x_0); \qquad  
T_{\theta,n}(x_n):=  \int S_{\theta,n}(x_{0:n})p_{\theta}(dx_{0:n-1}|y_{0:n-1},x_{n}), 
\quad n\ge 1.
\end{align*}
Notice that, first,
\begin{align*}
\mathcal{S}_{\theta,n}
%
  =\int T_{\theta,n}(x_{n})p_{\theta}(dx_{n}| y_{0:n}).
\end{align*}
We also have that, 
\begin{align}
T_{\theta,n}(x_{n}) &=\int\big[T_{\theta,n-1}(x_{n-1})+s_{\theta,n}(x_{n-1},x_{n})\big]p_{\theta}(dx_{n-1}|y_{0:n-1},x_{n}) \nonumber \\
&\equiv \frac{\int\big[T_{\theta,n-1}(x_{n-1})+s_{\theta,n}(x_{n-1},x_{n})\big]f_{\theta}(x_{n}|x_{n-1})p_{\theta}(
dx_{n-1}|y_{0:n-1})}{\int f_{\theta}(x_n|x_{n-1})p_{\theta}(dx_{n-1}|y_{0:n-1})},
\label{eq:recu}
\end{align}
-- see Proposition 2.1 of \citet{del2010forward} for a (simple) proof. 
This recursion provides an online -- forward-only -- 
advancement of FFBS for estimating the smoothing 
expectation of additive functionals. The complete method is summarised in  Algorithm \ref{alg:forward}: one key ingredient is that, during the recursion, 
values of the functional $T_{\theta,n}(x_n)$ 
are only required at the discrete positions $x_{n}^{(i)}$ determined by the forward particle filter.


In the SDE context, under Assumption \ref{ass:unknown}, the transition
density $f_{\theta}(\cdot|\cdot)$ is considered intractable, 
thus Algorithm \ref{alg:forward} -- apart from serving as a review of the method in \cite{del2010forward} -- does not appear to be practical in 
the continuous-time case.

\begin{algorithm}[!h]
\caption{Online Forward-Only Smoothing \citep{del2010forward}}
\label{alg:forward}
\begin{enumerate}
\item[(i)] Initialise particles $\{x_{0}^{(i)},W_{0}^{(i)}\}_{i=1}^{N}$, with  $x_{0}^{(i)}\stackrel{iid}{\sim}p_{\theta}(dx_0)$, 
$W_{0}^{(i)} \propto g_{\theta}(y_0|x_0^{(i)})$, and functionals 
$\widehat{T}_{\theta,0}(x_0^{(i)})= s_{\theta,0}(x_0^{(i)})$, 
for $1\le i\le N$. 
\item[(ii)] Assume that at time $n-1$, one has a particle approximation $\{x_{n-1}^{(i)},W_{n-1}^{(i)}\}_{i=1}^{N}$
of the filtering law $p_{\theta}(dx_{n-1}| y_{0:n-1})$ and  estimators $\widehat{T}_{\theta,n-1}(x_{n-1}^{(i)})$
of $T_{\theta,n-1}(x_{n-1}^{(i)})$, for $1\leq i\leq N$.
\item[(iii)] At time $n$, sample $x_{n}^{(i)}$, for $1\leq i \leq  N$, from the
mixture \citep{gordon1993novel},
\begin{align*}
x_{n}^{(i)}  \sim \widehat{p}_{\theta}(x_n|y_{0:n}) = \sum_{j=1}^{N}W_{n-1}^{(j)}f_{\theta}(x_{n}|x_{n-1}^{(j)}),
\end{align*}
and assign particle weights $W_n^{(i)}\propto g_{\theta}(y_n|x_n^{(i)})$, $1\le i \le N$.
\item[(iv)] Then set, for $1\leq i\leq N$, 
\begin{align*}
\widehat{T}_{\theta,n}(x_{n}^{(i)}) & =\frac{\sum_{j=1}^{N}W_{n-1}^{(j)}f_{\theta}(x_{n}^{(i)}| x_{n-1}^{(j)})}{\sum_{l=1}^{N}W_{n-1}^{(l)}f_{\theta}(x_{n}^{(i)}| x_{n-1}^{(l)})}\big[\widehat{T}_{\theta,n-1}(x_{n-1}^{(j)})+s_{\theta,n}(x_{n-1}^{(j)},x_{n}^{(i)})\big].
\end{align*}
\item[(v)] Obtain an estimate of $\mathcal{S}_{\theta,n}$ as, 
\begin{align*}
\widehat{\mathcal{S}}_{\theta,n} & =\sum_{i=1}^{N}W_{n}^{(i)}\widehat{T}_{\theta,n}(x_{n}^{(i)}).
\end{align*}
\end{enumerate}
\end{algorithm}

\section{Data Augmentation on Diffusion Pathspace}
\label{sec:SDE and Girsanov}

To overcome the intractability of the transition
density $f_{\theta}(\cdot|\cdot)$ of the SDE, we will work with an algorithm that is defined in 
continuous-time and makes use of the complete SDE path-particles in its development. 
The new method has connections with earlier works in the literature. \citet{sark:08} 
focus on the filtering problem for a class of models related to (\ref{eq:SDE})-(\ref{eq:trans}), and 
come up with an approach that requires the complete SDE path, for a limited class of diffusions with additive noise and no jumps. \cite{fearnhead2008particle} also deal with the filtering problem, and -- equipped with an unbiased estimator of the unknown transition density -- recast the 
problem as one of filtering over an augmented space that incorporates the randomness for the unbiased estimate. The method is accompanied by strict conditions on the drift and diffusion coefficient (one should be able to transform the SDE -- no jumps -- into one of unit diffusion coefficient; the drift of the SDE must have a gradient form). Our contribution requires, in principle, solely the diffusion coefficient invertibility Assumption \ref{ass:invert}; arguably, the weakened condition we require stems from the fact that our approach appears as the relatively 
most `natural' extension (compared to alternatives) of the standard discrete-time algorithm of \cite{del2010forward}. 

The latter discrete-time 
method requires the density  $f_\theta(x'|x)=f_\theta(dx'|x)/dx'$.
In continuous-time, we obtain
an analytically available Radon-Nikodym derivative of $p_{\theta}(d\mathbf{x}'|x)$, 
for a properly defined variate $\mathbf{x}'$ that involves information about the 
continuous-time path for moving from $x$ to $x'$ within time $\Delta$. 
We will give the complete algorithm in Section \ref{sec:cont}. 
In this section, 
we prepare the ground via carefully determining $\mathbf{x}'$ given $x$, and calculating the 
relevant densities to be later plugged in into our method.

\subsection{SDEs with Continuous Paths}
\label{subsec:Girsanov}

We work first with the process with continuous sample-paths, i.e.~of dynamics,
\begin{align}
dX_{t} =b_{\theta}(X_{t})dt+\sigma_{\theta}(X_{t})dW_{t}. \label{eq:SDEc}
\end{align}
We adopt an approach motivated by techniques used for MCMC
algorithms \citep{chib:04,goli:08,robe:01}.
Assume we are given starting point $x\in\mathbb{R}^{d_x}$, ending point $x'\in\mathbb{R}^{d_x}$,
and the complete continuous-time path  for the signal process in (\ref{eq:SDEc}) on 
$[0,T]$, for some $T\ge 0$. That is, we now work with the path process,
\begin{align}
\big[\,\{X_{t};t\in[0,T]\}\,\big|\, X_{0}=x,X_{T}=x'\,\big].
\label{target law}
\end{align}
Let $\mathbb{P}_{\theta,x,x'}$ denote the probability distribution of the pathspace-valued variable in (\ref{target law}).
We consider
the auxiliary bridge process $\tilde{X}=\{\tilde{X}_{t};t\in[0,T]\}$ defined as, 
\begin{align}
d\tilde{X}_{t} & =\Big\{ b_{\theta}(\tilde{X}_{t})+\frac{x'-\tilde{X}_{t}}{T-t}\Big\} dt+\sigma_{\theta}(\tilde{X}_{t})dW_{t},\quad \tilde{X}_{0}=x,\quad t\in[0,T],
\label{auxiliary SDE}
\end{align}
with corresponding probability distribution $\mathbb{Q}_{\theta,x,x'}$. Critically,
a path of $\tilde{X}$ starts at point $x$ and finishes at $x'$, w.p.~1.
Under regularity conditions, \citet{dely:06} prove that probability measures $\mathbb{P}_{\theta,x,x'}$, $\mathbb{Q}_{\theta,x,y}$ are absolutely continuous with respect to each other. We treat the auxiliary SDE (\ref{auxiliary SDE}) as a transform
from the driving noise to the solution, whence
a sample path, $X$, of the process 
$\tilde{X}=\{\tilde{X_t}; t\in[0,T]\}$, 
is produced by a mapping -- determined by (\ref{auxiliary SDE}) -- of a corresponding sample path, say $Z$, of the 
Wiener process. That is, we have set up a map, and -- under Assumption \ref{ass:invert} -- its inverse,   
\begin{equation} 
Z \mapsto X=: F_{\theta}(Z;x,x'), \quad Z =F^{-1}_{\theta}(X;x,x'). \label{transform}
\end{equation}
More analytically, $F_{\theta}^{-1}$ is given via the transform,
\begin{align*}
dZ_{t} & =\sigma_{\theta}(X_{t})^{-1}\Big\{dX_{t}-b_{\theta}(X_t)dt-\frac{x'-X_{t}}{T-t}dt\Big\}.
\end{align*}
In this case we define,
\begin{equation*}
\mathbf{x}' := (x',Z),
\end{equation*}
and the probability measure of interest is, 
\begin{equation}
\label{eq:double}
p_{\theta}(d\mathbf{x}'|x) := f_{\theta}(dx'|x)\otimes p_{\theta}(dZ|x',x). 
\end{equation}
%

%
%
%
\noindent 
Let $\mathbb{W}$ be the standard
Wiener probability measure on $[0,T]$. 
Due to the 1--1 transform, we have that, 
\begin{equation*}
\frac{p_{\theta}(dZ|x',x)}{\mathbb{W}(dZ)} \equiv 
\frac{d\mathbb{P}_{\theta,x,x'}}
{d\mathbb{Q}_{\theta,x,x'}}(F_{\theta}(Z;x,x')),
\end{equation*}
so it remains to obtain the density $d\mathbb{P}_{\theta,x,x'}/
d\mathbb{Q}_{\theta,x,x'}$.

\label{rem:density}
Such a Radom-Nikodym derivative has been object of interest in many works. \cite{dely:06} provided detailed conditions and a proof, but (seemingly) omit an expression for the normalising constant which is important in our case, as it involves the parameter $\theta$ -- in our applications later in the paper, we aim to infer about $\theta$.
\cite{papaspiliopoulos2013data, papaspiliopoulos2009importance} provide an expression based on a conditioning argument for the projection of the probability measures on $[0,t)$, $t<T$, and passage to the limit $t\uparrow T$. 
The derivations in \cite{dely:06} are extremely rigorous,
so we will make use the expressions in that paper. Following carefully the proofs of some of their main results (Theorem~5, together with Lemmas 7, 8) one can indeed retrieve the constant in the deduced density. In particular,
\cite{dely:06} impose the following conditions.
\begin{ass}
\label{ass:dely}
\begin{itemize}
\item[(i)] SDE (\ref{eq:SDEc}) admits a strong solution;
\item[(ii)] $v\mapsto\sigma_{\theta}(v)$ is in  $\mathcal{C}^{2}_b$ (i.e., twice continuously differentiable and bounded, with bounded first and second derivatives), and it is invertible, with bounded inverse;
\item[(iii)] $v\mapsto b_{\theta}(v)$ is locally Lipschitz, locally bounded.
\end{itemize}
\end{ass}
Under Assumption \ref{ass:dely},
\cite{dely:06} prove that,  
%
%
\begin{gather}
\frac{d\mathbb{P}_{\theta,x,x'}}
{d\mathbb{Q}_{\theta,x,x'}}(X) =  \frac{|\Sigma_{\theta}(x')|^{1/2}}{|\Sigma_{\theta}(x)|^{1/2}}
\times \frac{\mathcal{N}(x';x, T\Sigma_{\theta}(x))}
{f_{\theta}(x'|x)}\times \varphi_{\theta}(X;x,x'),
\label{eq:ddely}
\end{gather}
where $\mathcal{N}(v;\mu,V)$ is the density function of the Gaussian law on $\mathbb{R}^{d_x}$ 
with mean $\mu$, variance $V$, evaluated at $v$, 
$|\cdot|$ is matrix determinant, 
and $\varphi_{\theta}(X;x,x')$ is such that, 
\begin{align*}
\log \varphi_{\theta}(X;x,x') = \int_0^{T}&\big\langle \,b_{\theta}(X_t),\Sigma^{-1}_{\theta}(X_t)dX_t\,\big\rangle 
 - \tfrac{1}{2}\int_{0}^{T} \big\langle\, b_{\theta}(X_t),  \Sigma_{\theta}^{-1}(X_t)b_{\theta}(X_t)dt\,\big\rangle
\\ &\!\!\!\!\!\!\!- \tfrac{1}{2}\int_{0}^{T}
\tfrac{\left\langle\, (x'-X_t),\,d\Sigma^{-1}_{\theta}(X_t)(x'-X_t)\,\right\rangle}{T-t} 
-\tfrac{1}{2}\int_{0}^{T}\tfrac{\sum_{i,j=1}^{d_x}d\,[\,\Sigma_{\theta,ij}^{-1}, (x'_i-X_{t,i})(x'_j-X_{t,j})\,]}{T-t}.
\end{align*}
Here, $[\cdot,\cdot]$ denotes the quadratic variation process for semi-martingales; also, $\langle \cdot, \cdot \rangle$ is the standard inner-product on $\mathbb{R}^{d_x}$. We note that 
transforms different from \eqref{transform} have been proposed in the literature  \citep{dellaportas2006bayesian,kalogeropoulos2010inference}
to achieve the same effect of obtaining an 1--1 mapping of the latent path that has a density with respect to a measure that does not depend on $x$, $x'$ or $\theta$.
However, such methods are mostly applicable for scalar diffusions \citep{ait2008closed}. Auxiliary variables involving a random, finite selection of points of the latent path, based on 
the (generalised) Poisson estimator of \cite{fearnhead2008particle} are similarly
restrictive. In contrast to other attempts, our
methodology may be applied for a much more general class of SDEs, as determined by Assumption \ref{ass:invert} -- and further regularity conditions, e.g.~as in Assumption \ref{ass:dely}. Thus, continuing from (\ref{eq:double}) we have obtained that, 
\begin{align}
\nonumber
\frac{p_{\theta}(d\mathbf{x}'|x)}{(\mathrm{Leb}^{\otimes d_x}\otimes \mathbb{W})(d\mathbf{x}')}&= \varphi_{\theta}\big(F_{\theta}(Z;x,x');x,x'\big)\times
\mathcal{N}(x';x,T\Sigma_{\theta}(x))\times \frac{|\Sigma_{\theta}(x')|^{1/2}}{|\Sigma_{\theta}(x)|^{1/2}} \\
&=: p_{\theta}(\mathbf{x}'|x) \equiv p_{\theta}(\mathbf{x}'|x; T).
\label{eq:finally}
\end{align}
We have added the extra argument involving the length of path, $T$, in the last expression, as it will be of use in the next section. 

\begin{rem}
A critical point here is that the above density is \emph{analytically tractable}, thus by working on pathspace we have overcome the 
unavailability of the transition density $f_{\theta}(x'|x)$.
\end{rem}

\subsection{SDEs with Jumps}
\label{subsec:Jump-Diffusions}

We extend the above developments to
the more general case of the  $d_x$-dimensional jump diffusion model given in (\ref{eq:SDE}), which we re-write here for convenience,
\begin{align}
dX_{t} & =b_{\theta}(X_{t{-}})dt+\sigma_{\theta}(X_{t{-}})dW_{t}+dJ_{t},\quad X_0=x\in\mathbb{R}^{d_x},\,\,\,t\in[0,T].\label{Jump-SDE}
\end{align}
Recall that $J_{\theta}= J=\{ J_{t}\} $
denotes a compound Poisson process with jump intensity $\lambda_{\theta}(\cdot)$ and
jump-size density $h_{\theta}(\cdot)$.
Let $\mathbb{F}_{\theta,x}(\cdot)$ denote the law of the unconditional process
(\ref{Jump-SDE}) and $\mathbb{L}_{\theta}$ the law of the involved compound Poisson process.
%
%
We write $J=((\tau_1,b_1), \ldots, (\tau_{\kappa},b_{\kappa}))$
to denote the jump process, 
where $\{\tau_{i}\}$ are the times of events, $\{b_{i}\}$ the jump sizes
and $\kappa\ge 0$ the total number of events. 
In addition, we consider the reference 
measure $\mathbb{L}$, corresponding to unit rate Poisson process measure on $[0,T]$ multiplied with $\otimes_{i=1}^{\kappa+1} \mathrm{Leb}^{\otimes d_x}$.

\subsubsection*{Construct One}
We consider the random variate,
\begin{align*}
\mathbf{x}' = \big( J, \{x_{\tau_{i}-}\}_{i=1}^{\kappa+1}, \{Z(i)\}_{i=1}^{\kappa+1} \big),    
\end{align*}
under the conventions $x_{\tau_{0}-}\equiv x$, $x_{\tau_{\kappa+1}-}\equiv x'$, 
where we have defined,
\begin{align*}
Z(i) = F^{-1}_{\theta}(X(i)\,;\,x_{\tau_{i-1}},x_{\tau_{i}-}), \quad 
X(i) := \{X_t;t\in [x_{\tau_{i-1}}, x_{\tau_{i}})\},\quad 
1\le i \le \kappa+1.
\end{align*}
We have that, 
\begin{align*}
p_{\theta}(d\mathbf{x}'|x) := 
\mathbb{L}_{\theta}(dJ)\otimes \Big[ \otimes_{i=1}^{\kappa+1} 
\big\{\,f_{\theta}(dx_{\tau_{i}-}|x_{\tau_{i-1}})\otimes p_{\theta}(dZ(i)|x_{\tau_{i-1}},x_{\tau_{i}-})\,\big\}
\Big].
\end{align*}
Using the results about SDEs without jumps in Section \ref{subsec:Girsanov},
and in particular expression (\ref{eq:finally}), 
upon defining, 
\begin{equation*}
\mathbf{x}'(i):=(x_{\tau_{i}-}, Z(i)),    
\end{equation*}
we have that -- under an apparent adaptation of Assumption \ref{ass:dely},
\begin{align*}
\frac{f_{\theta}(dx_{\tau_{i}-}|x_{\tau_{i-1}})\otimes p_{\theta}(dZ(i)|x_{\tau_{i-1}},x_{\tau_{i}-})}{
\mathrm{Leb}^{\otimes d_x}(dx_{\tau_{i}-})\otimes \mathbb{W}(dZ(i)))
} = 
p_{\theta}(\mathbf{x}'(i)|x_{\tau_i-}; \tau_{i}-\tau_{i-1}),
\end{align*}
with the latter density $p_{\theta}(\mathbf{x}'(i)|x_{\tau_i-})$ 
determined as in (\ref{eq:finally}), given clear adjustments.
Thus, the density of $p_{\theta}(d\mathbf{x}'|x)$ with respect to the reference 
measure,
\begin{align*}
\mu(d\mathbf{x}'):=\mathbb{L}(dJ)\otimes \Big[  \otimes_{i=1}^{\kappa+1} \big\{ \mathrm{Leb}^{\otimes d_x}(dx_{\tau_{i}-})\otimes \mathbb{W}(dZ(i))\big\} \Big],
\end{align*}
is equal to, 
\begin{align*}
\frac{p_{\theta}(d\mathbf{x}'|x)}{\mu(d\mathbf{x}')}=
\frac{e^{-\int_{0}^{T}\lambda_{\theta}(t)dt}}{e^{-T}}
\cdot \prod_{i=1}^{\kappa}\big\{\lambda_{\theta}(\tau_i))h_{\theta}(b_i)\big\} 
\times 
\prod_{i=1}^{\kappa+1}
p_{\theta}(\mathbf{x}'(i)|x_{\tau_i-};\tau_i-\tau_{i-1}).
\end{align*}


%
%
%
%


\subsubsection*{Construct Two}

We adopt an idea used -- for a different problem -- in 
\cite{gonccalves2014exact}. 
Given $x,x'\in \mathbb{R}^{d_x}$, we define an auxiliary process $\tilde{X}_t$ as follows,
\begin{align}
d\tilde{X}_t & =\Big\{ b_{\theta}(\tilde{X}_t)+\frac{x'-J_{T}-\tilde{X}_t+J_t}{T-t}\Big\} dt+\sigma_{\theta}(\tilde{X}_t)dW_{t}+dJ_{t},\quad X_0=x,
\label{eq:auxJump}
\end{align}
so that $\tilde{X}_{T}=x'$, w.p.~1.
As with (\ref{transform}), we view (\ref{eq:auxJump})
as a transform, projecting a path, $Z$ of the Wiener process and the compound process, $J$,  onto a path, $X$, of the jump process. That is, we consider the 1--1 maps,
\begin{equation*} 
(J,Z) \mapsto X=: G_{\theta}(J,Z;x,x'), \quad (J,Z) =G^{-1}_{\theta}(X;x,x'). 
\end{equation*}
Notice that for the inverse transform, the $J$-part is obtained immediately, whereas for the $Z$-part one uses 
the expression
-- well-defined due to Assumption \ref{ass:invert} --,
\begin{align*}
dZ_{t} & =\sigma_{\theta}(X_{t})^{-1}\Big\{dX_{t}- dJ_t - b_{\theta}(X_t)dt-\frac{x'-J_T-X_{t}+J_t}{T-t}dt\Big\}.
\end{align*}
We denote by $\overline{\mathbb{P}}_{\theta,x,x'}$
the law of
the original process in (\ref{Jump-SDE}) conditionally 
on hitting $x'$ at time $T$. 
Also, we denote the distribution on pathspace 
induced by \eqref{eq:auxJump} as $\overline{\mathbb{Q}}_{\theta,x,x'}$.  
Consider the variate, 
\begin{align*}
\mathbf{x}' = ( x', J, Z ),    
\end{align*}
so that, 
\begin{equation*}
p_{\theta}(d\mathbf{x}'|x) := f_{\theta}(dx'|x)\otimes p_{\theta}\big(d(J,Z)|x',x\big). 
\end{equation*}
Due to the employed 1--1 transforms, we have that, 
\begin{align*}
\frac{p_{\theta}(d\mathbf{x}'|x)}
{\big(\mathrm{Leb}^{\otimes d_x}\otimes \mathbb{L}_{\theta}\otimes \mathbb{W}\big)(d\mathbf{x}') }  = f_{\theta}(x'|x)\times 
\frac{d\overline{\mathbb{P}}_{\theta,x,x'}}{d\overline{\mathbb{Q}}_{\theta,x,x'}}\big(G_{\theta}(J,Z;x,x')\big).
\end{align*}
Thus, using the parameter-free reference measure, 
\begin{align*}
\mu(d\mathbf{x}'):= \mathrm{Leb}^{\otimes d_x}\otimes \mathbb{L}\otimes \mathbb{W},
\end{align*}
one obtains that, 
\begin{align}
\frac{p_{\theta}(d\mathbf{x}'|x)}
{ \mu(d\mathbf{x}') }  &= f_{\theta}(x'|x)\times 
\frac{e^{-\int_{0}^{T}\lambda_{\theta}(t)dt}}{e^{-T}}
\cdot \prod_{i=1}^{\kappa}\big\{\lambda_{\theta}(\tau_i))h_{\theta}(b_i)\big\}\nonumber \\ 
& \qquad \qquad \qquad \qquad \qquad\qquad\qquad\times 
\frac{d\overline{\mathbb{P}}_{\theta,x,x'}}{d\overline{\mathbb{Q}}_{\theta,x,x'}}\big(F_{\theta}(J,Z;x,x')\big).
\label{eq:jjump}
\end{align}
\begin{rem}
\cite{dely:06} obtained the Radon-Nikodym derivative expression in  (\ref{eq:ddely}) after a great amount of rigorous analysis.
A similar development for the case of conditioned jump diffusions does not follow from their work, and can only be subject of dedicated research at the scale of a separate paper. This is beyond the scope of our work. In practice, one can proceed as follows. 
For grid size $M\ge 1$, and $\delta=T/M$, let  $\overline{\mathbb{P}}^{M}_{\theta,x,x'}(X_{\delta},\ldots, X_{(M-1)\delta}\,|\,X_0=x, X_{M\delta}=x')$
denote the time-discretised Lebesgue density of the $(M-1)$-positions of the 
conditioned diffusion with law $\overline{\mathbb{P}}_{\theta,x,x'}$.
Once (\ref{eq:jjump}) is obtained, a time-discretisation approach will give,
\begin{align*}
\overline{\mathbb{P}}^{M}_{\theta,x,x'}(X_{\delta},\ldots,& X_{(M-1)\delta}\,|\,X_0=x, X_{M\delta}=x')  \\
&=\frac{\overline{\mathbb{P}}_{\theta,x,x'}^{M}(X_{\delta},\ldots, X_{(M-1)\delta}, X_{M\delta}=x'\,|\,X_0=x)}{
\overline{\mathbb{P}}^{M}_{\theta,x,x'}
(X_{M\delta}=x'\,| X_0=x)
}.
\end{align*}
In this time-discretised setting, $f_{\theta}(x'|x)$ in (\ref{eq:ddely}) will be
replaced by $\overline{\mathbb{P}}^{M}_{\theta,x,x'}
(X_{M\delta}=x'\,| X_0=x)$. Thus, the intractable transition density over the complete time period will  cancel out, and one is left with an explicit expression to use on a PC. 
Compared to the method in Section \ref{subsec:Girsanov}, and the Construct One
in the current section, we do not have explicit theoretical evidence of a density on the pathspace. Yet, all numerical experiments we tried showed that the deduced algorithm was stable under mesh-refinement.  
We thus adopt the approach (or, conjecture) that the density in (\ref{eq:jjump}) exists -- under assumptions --, and can be obtained pending future research.
\end{rem}



\section{Forward-Only Smoothing for SDEs}
\label{sec:cont}

\subsection{Pathspace Algorithm}

We are ready to develop a forward-only particle smoothing 
method, under the scenario in (\ref{eq:exp})-(\ref{eq:sum}) on
the pathspace setting.
We will work with the pairs of random elements
\begin{align*}
(x_{k-1},\mathbf{x}_{k}), \quad 1\leq k\leq n,
\end{align*}
with $\mathbf{x}_{k}$ as defined in Section \ref{sec:SDE and Girsanov}, i.e.~containing pathspace elements,
and given by an 1--1 transform of $\{X_{s};s\in[t_{k-1},t_k]\}$, such that we can obtain 
a density for $p_{\theta}(\mathbf{x}_k|x_{k-1})$ with respect to 
a reference measure that does not involve $\theta$. Recall that  $p_{\theta}(d\mathbf{x}_k|x_{k-1})$ denotes the
probability law for the augmented variable $\mathbf{x}_k$ given 
$x_{k-1}$. 
We also write the corresponding density as, 
\begin{align*}
p_{\theta}(\mathbf{x}_k|x_{k-1}) := \frac{p_{\theta}(d\mathbf{x}_k|x_{k-1})}{\mu(d\mathbf{x}_k)}.
\end{align*}
The quantity of interest is now,
\begin{equation*}
\mathcal{S}_{\theta,n}=\int S_{\theta}(\mathbf{x}_{0:n})\,p_{\theta}(d\mathbf{x}_{0:n}|y_{0:n}), \quad n\ge 1,
\end{equation*}
for the class of additive functionals $S(\cdot)$ of the structure,
\begin{equation*}
 S_{\theta}(\mathbf{x}_{0:n})=\sum_{k=0}^{n}s_{\theta,k}(x_{k-1},\mathbf{x}_{k}),
\end{equation*} 
under the convention that $x_{-1}=\emptyset$.
Notice that we now allow $s_k(\cdot,\cdot)$ to be a function of 
$x_{k-1}$ and $\mathbf{x_k}$; thus, $s_k(\cdot,\cdot)$ can potentially correspond to integrals, or other pathspace functionals. 
We will work with a transition density on the enlarged space of $\mathbf{x}_k$.

Similarly to the discrete-time case in Section \ref{sec:online-smoothing}, we define the functional,
\begin{align*}
T_{\theta,n}(\mathbf{x}_{n}) & :=\int S_{\theta}(\mathbf{x}_{0:n})\,p_{\theta}(d\mathbf{x}_{0:n-1}| y_{0:n-1},\mathbf{x}_{n}).
\end{align*}
%
\begin{prop}
We have that,
\begin{align*}
\mathcal{S}_{\theta,n} & =\int T_{\theta,n}(\mathbf{x}_{n})\,p_{\theta}(d\mathbf{x}_{n}| y_{0:n}).
\end{align*}
\end{prop}
\begin{proof}
See Appendix \ref{proof1}
\end{proof}
\noindent Critically, as in \eqref{eq:recu}, we obtain the
following recursion.
\begin{prop} \label{prop:recursionT}
For any $n\geq1$, we have that, 
\begin{align*}
T_{\theta,n}(\mathbf{x}_{n}) & =\int\big[T_{\theta,n-1}(\mathbf{x}_{n-1})+s_{\theta,n}(x_{n-1},\mathbf{x}_{n})\big]\,p_{\theta}(d\mathbf{x}_{n-1}| y_{0:n-1},\mathbf{x}_{n})
\nonumber \\
&\equiv \frac{\int\big[T_{\theta,n-1}(\mathbf{x}_{n-1})+s_{\theta,n}(x_{n-1},\mathbf{x}_{n})\big]\,p_{\theta}(\mathbf{x}_{n}|x_{n-1})\,p_{\theta}(
d\mathbf{x}_{n-1}|y_{0:n-1})}{\int p_{\theta}(\mathbf{x}_n|x_{n-1})\,p_{\theta}(dx_{n-1}|y_{0:n-1})}.
\end{align*}
\end{prop}
\begin{proof}
See Appendix \ref{proof2}.
\end{proof}

\noindent Proposition \ref{prop:recursionT} gives rise to a Monte-Carlo methodology for a forward-only, online approximation of the smoothing expectation of interest. This is given in Algorithm \ref{alg:Forward-PF-infinite}.
\begin{algorithm}[!h]
\caption{Online Forward-Only Particle Smoothing on
Pathspace}
\label{alg:Forward-PF-infinite} 
\begin{enumerate}
\item[(i)]
Initialise the particles $\{x_{0}^{(i)},W_{0}^{(i)}\}_{i=1}^{N}$, with  $x_{0}^{(i)}\stackrel{i.i.d.}{\sim} p_{\theta}(dx_0)$, 
$W_{0}^{(i)} \propto g_{\theta}(y_0|x_0^{(i)})$, and functionals 
$\widehat{T}_{\theta,0}(x_0^{(i)})= s_{\theta,0}(\mathbf{x}_0^{(i)})$,
where $\mathbf{x}_0^{(i)}\equiv x_0^{(i)}$, $1\le i\le N$. 
\item[(ii)] Assume that at time $n-1$, one has a particle approximation $\{\mathbf{x}_{n-1}^{(i)},W_{n-1}^{(i)}\}_{i=1}^{N}$
of the filtering law $p_{\theta}(d\mathbf{x}_{n-1}| y_{0:n-1})$ and  estimators $\widehat{T}_{\theta,n-1}(\mathbf{x}_{n-1}^{(i)})$
of $T_{\theta,n-1}(\mathbf{x}_{n-1}^{(i)})$, $1\leq i\leq N$.
\item[(iii)] At time $n$, sample $\mathbf{x}_{n}^{(i)}$, for $1\leq i \leq  N$, from 
\begin{align*}
\mathbf{x}_{n}^{(i)}  \sim \widehat{p}_{\theta}(d\mathbf{x}_n|y_{0:n-1}) = \sum_{j=1}^{N}W_{n-1}^{(j)}p_{\theta}(d\mathbf{x}_{n}|x_{n-1}^{(j)}),
\end{align*}
and assign particle weights $W_n^{(i)}\propto g_{\theta}(y_n|y_{n-1},\mathcal{F}_n^{(i)})$, $1\le i \le N$.
\item[(iv)] Then set, for $1\leq i\leq N$, 
\begin{align*}
\widehat{T}_{\theta,n}(\mathbf{x}_{n}^{(i)}) & =\frac{\sum_{j=1}^{N}W_{n-1}^{(j)}\,p_{\theta}(\mathbf{x}_{n}^{(i)}| x_{n-1}^{(j)})}{\sum_{l=1}^{N}W_{n-1}^{(l)}\,p_{\theta}(\mathbf{x}_{n}^{(i)}| x_{n-1}^{(l)})}\,\big[\widehat{T}_{\theta,n-1}(\mathbf{x}_{n-1}^{(j)})+s_{\theta,n}(x_{n-1}^{(j)},\mathbf{x}_{n}^{(i)})\big].
\end{align*}
\item[(v)] Obtain an estimate of $\mathcal{S}_{\theta,n}$ as, 
\begin{align*}
\widehat{\mathcal{S}}_{\theta,n} & =\sum_{i=1}^{N}W_{n}^{(i)}\,\widehat{T}_{\theta,n}(\mathbf{x}_{n}^{(i)}).
\end{align*}
\end{enumerate}
\end{algorithm}

\begin{rem}
Algorithm \ref{alg:Forward-PF-infinite} uses a simple bootstrap filter with multinomial resampling applied at each step. 
The variability of the Monte-Carlo estimates can be further reduced by incorporating: more effective resampling (e.g., systematic resampling \citep{carpenter1999improved}, stratified resampling \citep{kitagawa1996monte});
dynamic resampling via use of Effective Sample Size; non-blind proposals in the propagation of the particles.

\end{rem}

\subsection{Pathspace versus Finite-Dimensional Construct}
One can attempt to define an algorithm without reference to the underlying pathspace. That is, in the case of no jumps (for simplicity)
an alternative approach can involve working with  a regular grid
on the period $[0,T]$, say $\{s_j=j\delta\}_{j=0}^{M}$, with $\delta=T/M$ for chosen size $M\ge 1$. Then, defining 
$\mathbf{x}'=(x_{\delta},x_{2\delta},\ldots, x_{M\delta})$,
and using, e.g., an Euler-Maruyama time-discretisation scheme to obtain the joint density of such an $\mathbf{x}'$ given $x=x_0$,
a Radon-Nikodym derivative, $p_{\theta}^{M}(\mathbf{x}'|x)$ on $\mathbb{R}^{M\times d_x}$, can be obtained with respect to the Lebesgue reference measure $\mathrm{Leb}^{\otimes (d_x\times M)}$, as a product of $M$ conditionally Gaussian densities. 
As shown e.g.~in the motivating example in the Introduction, such an approach would lead to estimates with variability that increases rapidly with $M$, for fixed Monte-Carlo particles $N$.
%
A central argument in this work is that one should develop the algorithm in a manner that respect the probabilistic properties of the SDE pathspace, before applying (necessarily) a time-discretisation for implementation on a PC. This procedure is not followed for purposes of mathematical rigour, but it has practical effects on algorithmic performance.

\subsection{Consistency}

For completeness, we provide an asymptotic result for Algorithm \ref{alg:Forward-PF-infinite} following standard results from the literature. We consider the following assumptions.

\begin{ass}
\label{assu: bound for G}
Let $\mathbf{X}$ and $\mathsf{X}$
denote the state spaces of $\mathbf{x}$ and $x$ respectively.
\begin{enumerate}
\item[(i)] For any relevant $y'$, $y$, $\mathcal{F}$, $x'$, we have that $g_{\theta}(y'| y,\mathcal{F})\equiv g_{\theta}(y'|x')$, 
where the latter 
is a positive function such that, for any $y$, $\sup_{x\in\mathsf{X}} g_{\theta}(y| x)<\infty$.
\item[(ii)] $\sup_{(x,\mathbf{x}')\in X\times \mathsf{\mathbf{X}}} p_{\theta}(\mathbf{x}'| x)<\infty$.
\end{enumerate}
\end{ass}

\begin{prop}
\label{prop:Under-,-for}
\begin{enumerate}
\item[(i)] Under Assumption \ref{assu: bound for G}, for any $n\geq0$, there exist constants
$b_{n},c_{n}>0$, such that for any $\epsilon>0$,
\begin{align*}
\mathrm{Prob}\big[\,|\mathcal{S}_{\theta,n}-\hat{\mathcal{S}}_{\theta,n}|>\epsilon\,\big] \leq b_{n}e^{-c_{n}N\epsilon^{2}}.
\end{align*}
\item[(ii)] For any $n\geq0$, $\hat{\mathcal{S}}_{\theta,n}\rightarrow\mathcal{S}_{\theta,n}$
w.p.1, as $N\rightarrow\infty$.
\end{enumerate}
\end{prop}
\begin{proof}
See Appendix \ref{proof3}.
\end{proof}

\section{Online Parameter/State Estimation for SDEs}
\label{sec:online}

In this section, we  derive an online gradient-ascent 
for partially observed SDEs. 
%
%
\citet{poyi:11} use the score function estimation methodology to
propose an online gradient-ascent algorithm for obtaining an
MLE-type parameter estimate, following ideas in \citet{legland1997recursive}.
In more detail, the method is based on the  Robbins-Monro-type of recursion,
\begin{align}
\theta_{n+1} & =\theta_{n}+\gamma_{n+1}\nabla\log p_{\theta_{0:n}}(y_{n}| y_{0:n-1})
\nonumber
 \\ &=\theta_{n}+\gamma_{n+1}\big\{\nabla\log p_{\theta_{0:n}}(y_{0:n})-\nabla\log p_{\theta_{0:n-1}}(y_{0:n-1})\big\}
\label{eq:recursive MLE}
\end{align}
where $\{\gamma_{k}\}_{k}$ is a positive decreasing sequence with,
\begin{align*}
\sum_{k=1}^{\infty}\gamma_{k}=\infty, \quad \sum_{k=1}^{\infty}\gamma_{k}^{2}<\infty.
\end{align*}
The meaning of quantity $\nabla\log p_{\theta_{0:n}}(y_{0:n})$ is that -- given a recursive method (in $n$) for 
the estimation of $\theta\mapsto\nabla\log p_{\theta}(y_{0:n})$
as we describe below and based on the methodology of Algorithm~\ref{alg:Forward-PF-infinite} -- one uses $\theta_{n-1}$ 
when incorporating $y_{n-1}$, then $\theta_n$ for $y_{n}$, 
and similarly for $k>n$. See \citet{legland1997recursive,tadic2018asymptotic} for analytical
studies of the convergence properties of the deduced algorithm, where 
under strong conditions the recursion
is shown to converge to the `true' parameter value, say $\theta^{\dagger}$, as $n\rightarrow\infty$.

Observe that, from Fisher's
 identity (see e.g.~\cite{poyi:11}) we have that,
\begin{align*}
\nabla\log p_{\theta}(y_{0:n}) & =\int\nabla\log p_{\theta}(\mathbf{x}_{0:n},y_{0:n})p_{\theta}(d\mathbf{x}_{0:n}|y_{0:n}).
\end{align*}
Thus, in the context of Algorithm \ref{alg:Forward-PF-infinite},  estimation of the score function corresponds to
the choice, 
\begin{align*}
S_{\theta}(\mathbf{x}_{0:n}) = \nabla\log p_{\theta}(\mathbf{x}_{0:n},y_{0:n})\equiv \sum_{k=0}^{n} \nabla\log p_{\theta}(\mathbf{x}_{k},y_{k}|x_{k-1}).
\end{align*}
and, 
\begin{align}
s_{\theta,k}(x_{k-1},\mathbf{x}_k) \equiv  \nabla\log p_{\theta}(\mathbf{x}_{k},y_{k}|x_{k-1}).
\label{eq:incscore}
\end{align}
Combination of the Robins-Morno recursion (\ref{eq:recursive MLE}) with the one in Algorithm \ref{alg:Forward-PF-infinite}, 
delivers Algorithm~\ref{alg:OnlineMLE}, which we have presented  here 
in some detail for clarity.

\begin{algorithm}[!h]
\caption{Online Gradient-Ascent for SDEs via 
Forward-Only Smoothing}
\label{alg:OnlineMLE}
\begin{enumerate}
\item[(i)] Assume that at time $n\ge 0$, one has current parameter estimate $\hat{\theta}_{n}$, 
 particle approximation $\{\mathbf{x}_{n}^{(i)},W_{n}^{(i)}\}_{i=1}^{N}$
of the `filter' $p_{\hat{\theta}_{0:n}}(d\mathbf{x}_{n}| y_{0:n})$ and estimators $\widehat{T}_{\hat{\theta}_{0:n},n}(\mathbf{x}_{n}^{(i)})$
of $T_{\hat{\theta}_{0:n},n}(\mathbf{x}_{n}^{(i)})$, for $1\leq i\leq N$.
\item[(ii)] Apply the iteration,
\begin{align*}
\hat{\theta}_{n+1} =\hat{\theta}_{n}+\gamma_{n+1}\big\{\nabla\log p_{\hat{\theta}_{0:n}}(y_{0:n})-\nabla\log p_{\hat{\theta}_{0:n-1}}(y_{0:n-1})\big\}
\end{align*}

\item[(iii)] At time $n+1$, 
sample $\mathbf{x}_{n+1}^{(i)}$, for $1\leq i \leq  N$, from the mixture, 
\begin{align*}
\mathbf{x}_{n+1}^{(i)}  \sim \widehat{p}_{\hat{\theta}_{n+1}}(d\mathbf{x}_{n+1}|y_{0:n}) = \sum_{j=1}^{N}W_{n}^{(j)}\,p_{\hat{\theta}_{n+1}}(\mathbf{x}_{n+1}|x_{n}^{(j)}),
\end{align*}
and assign particle weights $W_{n+1}^{(i)}\propto g_{\hat{\theta}_{n+1}}(y_{n+1}|y_n,\mathcal{F}_{n+1}^{(i)})$, $1\le i \le N$.
\item[(iii)] Then set, for $1\leq i\leq N$, 
\begin{align*}
\widehat{T}_{\hat{\theta}_{0:n+1}, n+1}(\mathbf{x}_{n+1}^{(i)}) & =\frac{\sum_{j=1}^{N}W_{n}^{(j)}\,p_{\theta}(\mathbf{x}_{n+1}^{(i)}| x_{n}^{(j)})}{\sum_{l=1}^{N}W_{n}^{(l)}\,p_{\theta}(\mathbf{x}_{n+1}^{(i)}| x_{n}^{(l)})}\,\big[\widehat{T}_{\hat{\theta}_{0:n},n}(\mathbf{x}_{n}^{(j)})+s_{\theta,n}(x_{n}^{(j)},\mathbf{x}_{n+1}^{(i)})\big],
\end{align*}
where, on the right-hand-side we use the parameter, 
\begin{align*}
\theta = \hat{\theta}_{n+1}.
\end{align*}
\item[(iv)] Obtain the estimate, 
\begin{align*}
\widehat{\mathcal{S}}_{\hat{\theta}_{n+1},n+1} & =\sum_{i=1}^{N}W_{n+1}^{(i)}\,\widehat{T}_{\hat{\theta}_{0:n+1}, n+1}(\mathbf{x}_{n+1}^{(i)}).
\end{align*}

\end{enumerate}
\end{algorithm}

\begin{rem}
When the joint density of $(\mathbf{x}_{0:n},y_{0:n})$ is in the exponential family, an online EM algorithm can also be developed;
see \cite{del2010forward} for the discrete-time case.
\end{rem}

\section{Numerical Applications}
\label{sec:numerics}
All numerical experiments shown in this section were obtained on a standard PC, with code written in \texttt{Julia}. Unless otherwise stated, executed algorithms discretised the pathspace using the Euler--Maruyama scheme with $M=10$ imputed points between consecutive observations, the latter assumed to be separated by 1 time unit. Algorithmic costs scale linearly with $M$.

To specify the schedule for the scaling parameters $\{\gamma_{k}\}$
in \eqref{eq:recursive MLE}, we use the standard adaptive method termed \texttt{ADAM}, by \cite{kingma2014adam}. Assume that after $n$ steps, we have 
$c_{n}:=-(\nabla\log p_{\hat{\theta}_{0:n}}(y_{0:n})-\nabla\log p_{\hat{\theta}_{0:n-1}}(y_{0:n-1}))$.
\texttt{ADAM} applies the following iterative procedure,
\begin{gather*}
m_{n}=m_{n-1}\beta_{1}+(1-\beta_{1})c_{n}, \quad  v_{n}=v_{n-1}\beta_{2}+(1-\beta_{2})c_{n}^{2},\\
\hat{m}_{n}=m_{n}/(1-\beta_{1}^{n}),\quad \hat{v}_{n}=v_{n}/(1-\beta_{2}^{n}),\\
\hat{\theta}_{n+1}=\hat{\theta}_{n}-\alpha\hat{m}_{n}/(\sqrt{\hat{v}_{n}}+\epsilon),
\end{gather*}
where $(\beta_{1},\beta_{2},\alpha,\epsilon)$ are tuning parameters.
Convergence properties of \texttt{ADAM} have been
widely studied \citep{goodfellow2016deep,kingma2014adam,reddi2019convergence}. Following \citet{kingma2014adam}, in all cases  
below we set   $(\beta_{1},\beta_{2},\alpha,\epsilon)=(0.9,0.999,0.001,10^{-8})$.
\texttt{ADAM} is nowadays a standard and very effective addition to the type of recursive inference algorithms we are considering here, even more so as for increasing dimension of unknown parameters. See the above references for more motivation
and details.


%
\subsection{Ornstein-Uhlenbeck SDE}
We consider the model,
\begin{gather}
dX_{t}  =\theta_{1}(\theta_2-X_{t})+\theta_{3}dW_{t}+dJ_t,\quad X_0=0,\nonumber \\
y_{i}=x_{i}+\epsilon_{i}, \quad i\ge 1, \quad \epsilon_{i}\stackrel{i.i.d.}{\sim}\mathcal{N}(0,0.1^{2}), \quad \Delta=1,
\label{eq:stability}
\end{gather}
where $J_{t}=\sum_{i=1}^{N_{t}}\xi_{i}$, 
$N_{t}$ is a Poisson process of intensity $\lambda\ge 0$,
and $\xi_{i}\overset{i.i.d.}{\sim}\mathcal{U}\left(-\zeta,\zeta\right)$, 
$\zeta> 0$. 

\textbf{O--U: Experiment 1.} 
Here, we choose $\lambda=0$, so the score function is
analytically available; we also fix $\theta_2=0$. 
We simulated datasets of size $n= 2,500, 5,000, 7,500, 10,000$ under true remaining parameter values $(\theta_1^{\dagger},\theta_3^{\dagger}) =(0.4,0.5)$.
We executed Algorithm \ref{alg:Forward-PF-infinite}
to approximate the score function with  $N=50, 100, 150$ particles,
for all above datasets.
Fig.~\ref{fig:stability} summarises results of $R=50$ replicates 
of estimates of the bivariate score function, with the dashed lines showing the true values.

%
%

\textbf{O--U: Experiment 2.} 
We now present two scenarios: 
one still without jumps (i.e., $\lambda=0$); and one with jumps, such that we 
fix $\lambda=0.5$, $\zeta = 0.5$.
In both scenarios, the parameters to be estimated are $\theta=(\theta_1,\theta_2,\theta_3)$.
Beyond the mentioned fixed values for $\lambda$, $\zeta$, 
in both scenarios we generated $n=20,000$ data points using the true parameters values
$\theta^{\dagger}=(\theta_{1}^{\dagger},\theta_{2}^{\dagger},\theta_{3}^{\dagger})=(0.2,0.0,0.2)$.
We applied Algorithm \ref{alg:OnlineMLE} with $N=100$ particles, and initial parameter values $\hat{\theta}_0=(1.0,1.0,1.0)$ for both scenarios. We used Construct 1 for the scenario with jumps.
Fig.~\ref{fig:OUresult} summarises the results. 

\textbf{O--U: Experiment 3.} 
To compare the efficiency of Constructs 1 and 2 for the case of jump processes, as developed in Section \ref{subsec:Jump-Diffusions},
we applied Algorithm \ref{alg:Forward-PF-infinite} -- for each case -- 
to approximate the score function (at the true values for $(\theta_{1},\theta_{2},\theta_{3})$), with $N=50,100,150,200$ particles,
with $n=10$ data generated from O--U process \eqref{eq:stability} 
with fixed $\lambda=0.5$, $\zeta=0.5$, and true remaining parameters
$(\theta_{1}^{\dagger},\theta_{2}^{\dagger},\theta_{3}^{\dagger})=(0.3,0.0,0.2)$.
Fig.~\ref{fig:Boxplots-of-estimated} shows boxplots summarising the results 
from $R=50$ replications of the obtained estimates.
We note that performance of Construct 2 is, in this 
case, much better than that of Construct 1.

\begin{figure}[!h]
\begin{centering}
\vspace{-0.3cm}
\includegraphics[scale=0.6]{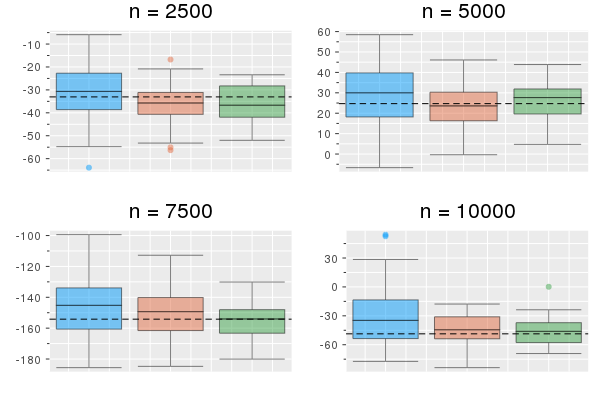}
\par\end{centering}
\vspace{-0.6cm}
\caption{\emph{O--U: Experiment 1:} Boxplots of estimated score functions of $\theta_{1}$ at  \mbox{$\theta_1=0.4$}, over $R=50$ experiment replications, 
for the O--U model \eqref{eq:stability}, without jumps ($\lambda=0$), fixed $\theta_2=0$, and free parameters $(\theta_1, \theta_3)$, with
true values
 $(\theta_{1}^{\dagger},\theta_{3}^{\dagger})=(0.4,0.5)$ used for data generation.
The left, middle and right boxplots in each panel correspond to values $N=50, 100, 150$, respectively. The horizontal dashed lines show the correct values $(-33.1,24.7,-154.2,-48.7)$
for datasizes $n=2,500, 5,000, 7,500, 10,000$, respectively.}
\label{fig:stability}
\end{figure}

\begin{figure}[!h]
\begin{centering}
\includegraphics[scale=0.6]{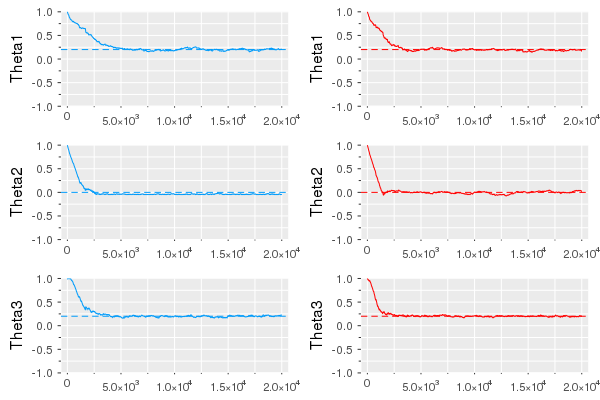}
\par\end{centering}
\caption{\emph{O--U: Experiment 2:} Trajectories from execution of Algorithm \ref{alg:OnlineMLE} for the estimation of parameters $(\theta_1,\theta_2,\theta_3)$
of the O--U model (\ref{eq:stability}).
We used 
 $N=100$ particles and initial value $\hat{\theta}_0=(1.0,1.0,1.0)$.
\emph{Left Panel:} results for O--U model with
jumps ($\lambda=0.5$, $\zeta=0.5$). \emph{Right Panel:} Results for O--U model without jumps ($\lambda=0$). The
horizontal dashed lines in the plots show the true parameter values
$(\theta_{1}^{\dagger},\theta_{2}^{\dagger},\theta_{3}^{\dagger})=(0.2,0.0,0.2)$
used in both scenarios.}
\label{fig:OUresult}
\end{figure}

\begin{figure}[!h]
\begin{centering}
\includegraphics[scale=0.6]{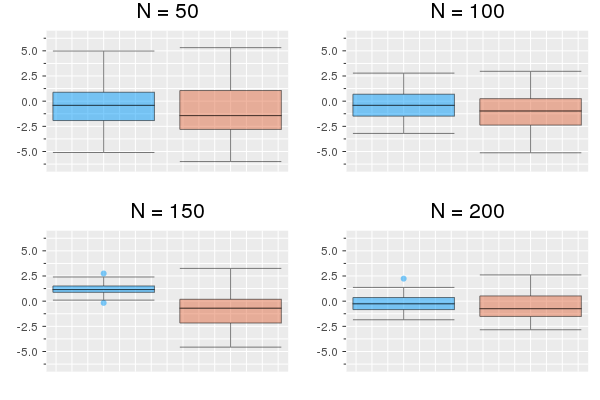}
\par\end{centering}
\vspace{-0.3cm} 
\caption{\emph{O--U: Experiment 3:} Boxplots of estimated score function
of $\theta_{1}$, at $\theta_1=0.3$, over $R=50$ experiment replications, for O--U model \eqref{eq:stability},
with fixed $\lambda=0.5$, $\zeta=0.5$, and true parameters 
$(\theta_{1}^{\dagger},\theta_{2}^{\dagger},\theta_{3}^{\dagger})=(0.3,0.0,0.2)$,
used to generate $n=10$ observations. Boxplots on the right (left) of each panel correspond to 
Construct~1 (Construct~2).}
\label{fig:Boxplots-of-estimated}
\end{figure}

\subsection{Periodic Drift SDE}
\label{subsec:per}

We consider the (highly) nonlinear
model,
\begin{gather}
dX_{t} =\sin\left(X_{t}-\theta_{1}\right)dt+\theta_{2}dW_{t},\quad X_0=0,
\nonumber\\
y_{i}=x_i+\epsilon_{i}, \quad i\ge 1, \quad \epsilon_{i}\stackrel{i.i.d.}{\sim}\mathcal{N}(0,0.1^{2}), \quad \Delta=1, 
\label{eq:sine process}
\end{gather}
for $\theta_{1}\in[0,2\pi)$, $\theta_2>0$.
We used true values $(\theta_{1}^{\dagger},\theta_{2}^{\dagger})=(\pi/4,0.9)$,  
generated $n=10^4$ data, and applied Algorithm \ref{alg:OnlineMLE} with $N=100$ particles and initial
values $(\hat{\theta}_{1,0}, \hat{\theta}_{2,0}) =(0.1,2)$. To illustrate the mesh-free performance of Algorithm \ref{alg:OnlineMLE}, we choose two values, $M=10$ and $M=100$ and compare the obtained numerical results. Figure \ref{fig:Sine-result} shows the data (Left Panel) and the trajectory of the estimated parameter values (Right Panel). Note that the SDE in (\ref{eq:sine process}) is making infrequent jumps between a number of modes.  Indeed, in this case, the results are stable when  increasing $M$ from $10$ to $100$.

\begin{figure}[!h]
\centering{}
\vspace{-0.6cm}
\includegraphics[scale=0.7]{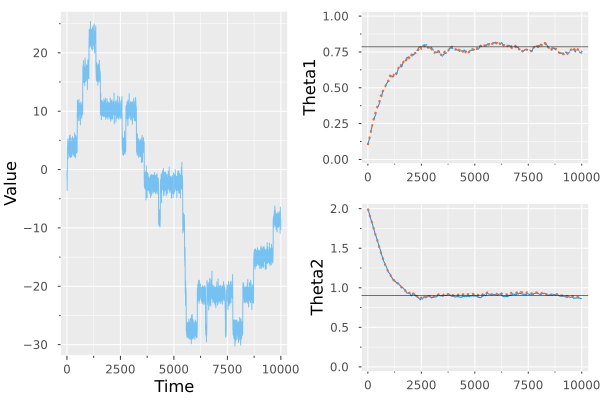}
\vspace{-0.3cm}
\caption{\emph{Left Panel:} Data generated from 
model (\ref{eq:sine process}) for true values
$\theta^{\dagger}=(\pi/4,0.9)$. \emph{Right Panel:}
Trajectories obtained  from Algorithm \ref{alg:OnlineMLE},
with 
$N=100$ and initial values $(\hat{\theta}_{1,0},\hat{\theta}_{2,0})=(0.1,2)$. Solid lines show the results for $M=10$ and dotted lines for $M=100$.}
\label{fig:Sine-result}
\end{figure}

%

\subsection{Heston Model}

Consider the Heston model \citep{heston1993closed},
for price $S_t$ and volatility $X_t$, modelled as,  
\begin{align*}
dS_{t} &=\theta_{4}S_{t}dt+S_{t}X_{t}^{1/2}dB_{t}, \\
dX_{t}  &=\theta_{1}(\theta_{2}-X_{t})dt+\theta_{3}X_{t}^{1/2}dW_{t},\quad 
X_0= \theta_2,
\end{align*}
where $B_t$, $W_t$ are independent standard Brownian motions. 
Define the log-price
$U_{t}=\log(S_{t})$, so that application of Itô's lemma gives, 
\begin{align}
dU_{t}&=(\theta_{4}-\tfrac{X_{t}}{2})dt+X_{t}^{1/2}dB_{t},\nonumber\\
dX_{t}&=\theta_{1}(\theta_{2}-X_{t})dt+\theta_{3}X_{t}^{1/2}dW_{t},\quad X_0=\theta_2.
\label{eq:Heston model}
\end{align}
We make the standard assumption $2\theta_{1}\theta_{2}>\theta_{3}^{2}$, 
so that the CIR process $X=\{X_t\}$ will not hit $0$; 
also, we have $\theta_1, \theta_2, \theta_3, \theta_4>0$. Process $U_{t}$ is observed at discrete times,
so that $y_{i}=U_{t_{i}}$, $i\ge 1$. Thus, we have, 
\begin{align*}
\big[\,y_{i}\,\big|\,y_{i-1},\,\{X_{s};s\in[t_{i-1},t_{i}]\}\,\big] & \sim\mathcal{N}(y_{i};\mu_{i},\Sigma_{i}),
\end{align*}
where we have set, 
\begin{align*}
\mu_{i}=y_{i-1}+\int_{t_{i-1}}^{t_{i}}( \theta_{4}-\tfrac{X_s}{2}) ds,\quad
\Sigma_{i}&=\int_{t_{i-1}}^{t_{i}}X_{s}ds.
\label{eq:likeli_SV}
\end{align*}
We chose true parameter value $\theta^{\dagger}=(0.1,1.0,0.2,0.45)$,
and generated $n=10^4$ observations, with $\Delta=1$.
We applied Algorithm \ref{alg:OnlineMLE} with $\hat{\theta}_0=(0.005,0.1,0.4,0.3)$ and $N=100$ particles.
The trajectories of the estimated parameters are given in Figure \ref{fig:heston_result}. 

\begin{figure}[!h]
\begin{centering}
\includegraphics[scale=0.6]{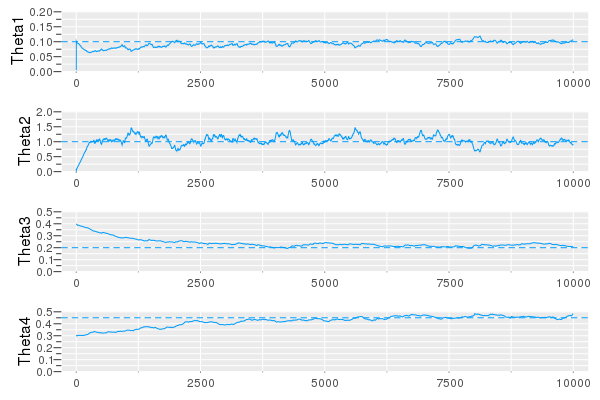}
\par\end{centering}
\vspace{-0.3cm}
\caption{Parameter estimate trajectories produced by Algorithm \ref{alg:OnlineMLE}, for data generated from Heston model \eqref{eq:Heston model}. The algorithm used initial values $\hat{\theta}_0=(0.005,0.1,0.4,0.3$), and $N=100$ particles.
The horizontal dashed lines indicate the true parameter values $\theta^{\dagger}=(0.1,1.0,0.2,0.45)$.}
\label{fig:heston_result}
\end{figure}

\subsection{Sequential Model Selection -- Real Data}
We use our methodology to carry out sequential model selection, motivated by applications in \cite{eraker2003impact,johannes2009optimal}.
Recall that BIC \citep{schwarz1978estimating} for model $\mathcal{M}$, with parameter vector $\theta$, and data $y_{1:n}$
is given by,
\begin{align*}
BIC(\mathcal{M}) & :=-2\ell_{\hat{\theta}_{MLE}}(y_{1:n})+
\mathrm{dim}(\theta)\log n,
\end{align*}
where $\hat{\theta}_{MLE}$ denotes the MLE, 
and $\ell(\cdot)$ the log-likelihood.
BIC and the (closely related) Bayes Factor are known to have good asymptotic properties, e.g.~they are consistent Model Selection Criteria, under the assumption of model-correctness, in the context of nested models, for specific classes of models (see e.g.~\cite{chib2016bayes, nishii1988maximum,sin1996information}, and \cite{yonekura2018asymptotic} for the case of discrete-time nested 
HMMs). One can plausibly conjecture such criteria will also perform well
for the type of continuous-time HMMs we consider here. See also \citet{eguchi2018schwarz} for
a rigorous analysis of BIC for diffusion-type models.
Given models $\mathcal{M}_a, \mathcal{M}_b$, with $\mathcal{M}_a\subset \mathcal{M}_b$, 
we use the difference, 
\begin{align*}
BIC(\mathcal{M}_{ab}):=-2(\ell_{\hat{\theta}_{MLE,a}}(y_{1:n})+
\ell_{\hat{\theta}_{MLE,b}}(y_{1:n}))+
(\mathrm{dim}(\theta_{a})-\mathrm{dim}(\theta_{b}))\log n,
\end{align*}
with involved terms defined as obviously, to choose between $\mathcal{M}_a$ and $\mathcal{M}_b$.

We note that 
\citet{eraker2003impact} use Bayes Factor for model selection 
in the context of SDE processes; that work 
uses MCMC to estimate parameters, and is not sequential (or online). 
We stress that our methodology allows for carrying out inference for parameters and the signal part of the model, and simultaneously allowing for model selection, all in an online fashion. Also, it is worth noting that \citet{johannes2009optimal} use
the sequential likelihood ratio for model comparison, but such quantity
can overshoot, i.e.~the likelihood ratio will tend to
choose a large model. 
Also, that work uses fixed calibrated
parameters, so it does not relate to online parameter inference. 



\begin{rem}
\label{rem:proxy}
We will use the running estimate of the parameter vector as a proxy for the
MLE given the data already taken under consideration. 
Similarly, we use the weights of the particle filter to obtain a running 
proxy of the log-likelihood evaluated at the MLE, thus overall 
an online approximation of BIC. 
Such an approach, even not fully justified theoretically,  
can provide reasonable practical insights when performing model comparison 
in an online manner -- particularly so, given when there is no alternative, to the best of our knowledge, for such an objective.
\end{rem}


We consider the following family of nested SDE models,
\begin{align*}
dX_{t} & =b_{\theta}^{(i)}(X_{t})+\theta_{4}\sqrt{X_{t}}dW_{t},
\end{align*}
where,
\begin{align*}
\mathcal{M}_{1}: &\quad  b_{\theta}^{(1)}(X_{t})=\theta_{0}+\theta_{1}X_{t},\\
\mathcal{M}_{2}: &\quad  b_{\theta}^{(2)}(X_{t})=\theta_{0}+\theta_{1}X_{t}+\theta_{2}^{2}X_{t}^{2},\\
\mathcal{M}_{3}: &\quad  b_{\theta}^{(3)}(X_{t})=\theta_{0}+\theta_{1}X_{t}+\theta_{2}^{2}X_{t}^{2}+\frac{\theta_{3}}{X_{t}}.
\end{align*}
Such models have been used for short-term interest rates; see \citet{jones2003nonlinear,durham2003likelihood,ait1996testing,bali2006comprehensive}
and references therein for more details. Motivated by
\cite{dellaportas2006bayesian, stanton1997nonparametric}, we applied
our methodology to daily 3-month Treasury Bill rates,
from 2nd of January 1970, to 29th of December 2000, providing $n=7,739$ observations.
The data can be 
obtained from the Federal Reserve Bank of St.~Louis, at webpage \url{https://fred.stlouisfed.org/series/TB3MS}. The dataset is shown
at Fig.~\ref{fig:-3-month-treasury}.

\begin{figure}[!h]
\begin{centering}
\includegraphics[scale=0.6]{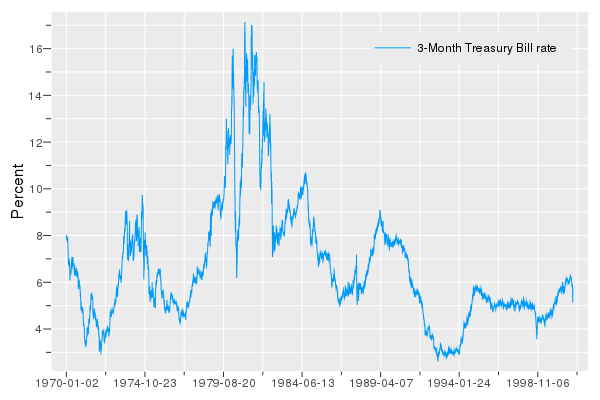}
\par\end{centering}
\vspace{-0.4cm}
\caption{Daily values of 3-Month Treasury
Bill rates from 2 Jan 1970 to 29 Dec 2000.}
\label{fig:-3-month-treasury}
\end{figure}

Fig.~\ref{fig:BICex1} shows the BIC differences, obtained online, for each of the three pairs of models. 
The results suggest that $\mathcal{M}_{2}$ does not
provide a good fit to the data -- relatively to $\mathcal{M}_{1}$, $\mathcal{M}_{3}$ -- for almost all period under consideration. 
In terms of $\mathcal{M}_{1}$ against $\mathcal{M}_{3}$, there is evidence 
that once the data from around 1993 onwards are taken under consideration, $\mathcal{M}_{3}$ is preferred to $\mathcal{M}_{1}$.
In general, one can claim that models with non-linear drift should be taken under consideration for fitting daily 3-month Treasury Bill rates,
without strong evidence in favour or against linearity. 
This non-definite conclusion is in some agreement with the empirical
studies in \citet{chapman2000short,durham2003likelihood,dellaportas2006bayesian}.

\begin{figure}[!h]
\noindent \begin{centering}
\includegraphics[scale=0.6]{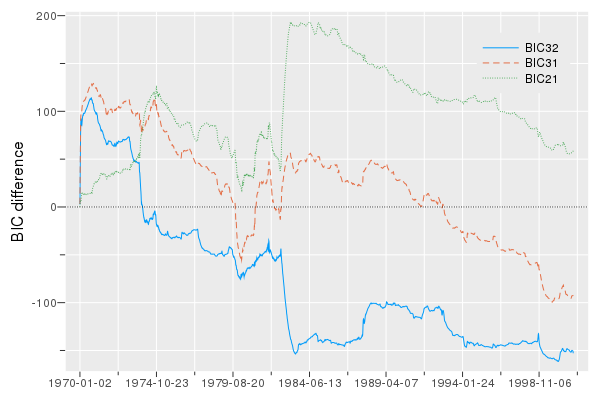}
\par\end{centering}
\vspace{-0.2cm}
\caption{Online estimation of BIC differences. The solid line stands
for $BIC(\mathcal{M}_{32})$, the dashed line for $BIC(\mathcal{M}_{31})$,
and the dotted line for $BIC(\mathcal{M}_{21})$.}
\label{fig:BICex1}
\end{figure}

\section{Conclusions and Future Directions}
\label{sec:Conclusion-and-remarks}

We have introduced an online particle smoothing methodology for discretely observed (jump) diffusions with intractable transition densities. Our approach overcomes 
such intractability by formulating the problem on
pathspace, thus delivering an algorithm that -- besides regulatory conditions -- requires only the 
weak invertibility Assumption \ref{ass:invert}. Thus, we have covered 
a rich family of SDE models, when related literature  
imposes quite strong restrictions. Combining our online
smoothing algorithm with a Robbins-Monro-type approach of Recursive Maximum-Likelihood, we set up an online stochastic gradient-ascent for the likelihood function of the SDEs under consideration. 
The algorithm provides a wealth of interesting output, that can provide a lot of useful insights in statistical applications. The numerical examples show a lot of promise for the performance of the methodology.
Our framework opens up a number of routes for insights and future research, 
including the ones described below.
\begin{itemize}
\item[(i)] In the case of SDEs of jumps, we have focused on jump dynamics driven by compound Poisson processes. There is great scope for generalisation here, and one can extend the algorithm to different cases of jump processes, also characterised by more complex dependencies between the jumps and the paths of the solution of the SDE, $X=\{X_t\}$.  
Extensions to time-inhomogeneous cases with $b(v)=b(t,v)$, $\sigma(v) = \sigma(t,v)$,
are immediate; we have chosen the time-homogeneous models only for purposes of  presentation simplicity. 
\item[(ii)] Since the seminal work of \cite{dely:06}, more `tuned' auxiliary bridge processes have appeared in the literature, see e.g.~the works of \cite{schauer2017guided, van2017bayesian}. Indicatively, the work in \cite{schauer2017guided} considers bridges of the form (in one of the many options they consider),
\begin{equation}
d\tilde{X}_t = \Big\{b(\tilde{X}_t) + \Sigma(\tilde{X}_t)\Sigma^{-1}(x')\frac{x'-\tilde{X}_t}{T-t}    \Big\}dt + 
\sigma(\tilde{X}_t)dW_t.
\label{eq:sch}
\end{equation}
Auxiliary bridge processes that are closer in dynamics to the diffusion bridges of the given signal are expected to reduce the variability of Monte-Carlo algorithm, thus progress along the above direction can be immediately incorporated in our methodology and improve its performance. For instance, as noted in \cite{schauer2017guided}, use of (\ref{eq:sch}) will give a Radon-Nikodym derivative where stochastic integrals cancel out. Such a setting is known to considerably reduce the variability of Monte-Carlo methods, see e.g.~the numerical examples in \cite{durh:02} and the discussion in \cite[Section 4]{papaspiliopoulos2009importance}.

\item[(iii)] The exact specification of the recursion used for the online estimation of unknown parameters is in itself an problem of intensive research in the field of stochastic optimisation. One would ideally aim for the recursion procedure to provide parameter estimates which are as close to the unknown parameter as the data (considered thus far) permit. In our case, we have used a fairly `vanilla' recursion, maybe with the exception of the \texttt{Adam} variation. E.g., recent works in the Machine Learning community have pointed at the use of `velocity' components in the recursion to speed up convergence,
see, e.g., \cite{suts:13, yuan:16}.
\item[(iv)] 
We mentioned in the main text several modifications that can
improve algorithmic performance: dynamic resampling, stratified resampling, non-blind proposals in the filtering steps, choice of auxiliary processes. Parallelisation and use of  
HPC are obvious additions in this list.

\item[(v)] 
We stress that the algorithm involves a filtering step, and a step 
that approximates the values of the instrumental function. 
These two procedures should be thought of separately.  
One reason for including two approaches in the case of jump diffusions (Constructs 1 and 2) is indeed to highlight this point. The two Constructs are identical in terms of the filtering part.
Construct 1 incorporates in $\mathbf{x}'$ the location of the path at all times of jumps; thus, when the algorithm `mixes' all pairs of $\{x_{k-1}^{(j)}\}$,  $\{\mathbf{x}_{k}^{(i)}\}$,  at the update of the instrumental function (see Step (iv) of Algorithm \ref{alg:Forward-PF-infinite}), many of such pairs can be incompatible (simply consider any such pair with $i\neq j$, when $\{\mathbf{x}_{k}^{(i)}\}$ has in fact been generated conditionally on $\{x_{k-1}^{(i)}\}$). Such an effect is even stronger in the case of the standard algorithm applied in the motivating example in the Introduction, and partially explains the inefficiency of that algorithm. In contrast, in Construct 2, $\mathbf{x}'$ contains less information about the underlying paths, 
thereby improving the compatibility of pairs selected from  particle populations $\{x_{k-1}^{(j)}\}$,  $\{\mathbf{x}_{k}^{(i)}\}$,
thus -- not surprisingly -- Construct 2 seems more effective 
than Construct~1. 
\item[(vi)]
The experiments included in Section 6 of the paper aim to highlight the \emph{qualitative} aspects of the introduced approach. 
The key property of robustness to mesh-refinement underlies all numerical applications, and is explicitly illustrated  in the experiment in Section 6.2, and within Figure~6.4.
In all cases, the code was written in \texttt{Julia}, was not parallelised or optimised and was run on a PC, with executions times that were quite small due to the online nature of the algorithm (the most expensive experiment took about 30 minutes).
We have avoided a comparison with potential alternative approaches as we believe this would obscure the main messages we have tried to convey within the numerical results section. Such a detailed numerical study can be the subject of future work.
\end{itemize}

\section*{Acknowledgments}
We are grateful to the AE and the referees for several useful comments that helped us to improve the content of the paper.


\bibliographystyle{apalike}

\bibliography{refereces}

\appendix

\section{Proof of Proposition 1} \label{proof1}
\begin{proof}
We have the integral, 
\begin{align}
\int T_{\theta,n}(\mathbf{x}_{n})\,p_{\theta}(d\mathbf{x}_{n}| y_{0:n}) =  
\int S_{\theta}(\mathbf{x}_{0:n})\,p_{\theta}(d\mathbf{x}_{0:n-1}| y_{0:n-1},\mathbf{x}_{n})\otimes 
p_{\theta}(d\mathbf{x}_{n}| y_{0:n}).
\label{eq:iinf}
\end{align}
Also, simple calculations give, 
\begin{align*}
p_{\theta}(d\mathbf{x}_{0:n-1}| y_{0:n-1},&\mathbf{x}_{n})
\otimes p_{\theta}(\mathbf{x}_{n}| y_{0:n}) \\ 
&\equiv  p_{\theta}(d\mathbf{x}_{0:n-1}| y_{0:n},\mathbf{x}_{n})\otimes p_{\theta}(\mathbf{x}_{n}| y_{0:n})
= p_{\theta}(d\mathbf{x}_{0:n}| y_{0:n}).
\end{align*}
Using this expression in the integral on the right side of (\ref{eq:iinf}) completes the proof.
\end{proof}

\section{Proof of Proposition 2} \label{proof2}
\begin{proof}
Simply note that,
\begin{align*}
\int T_{\theta,n-1}(\mathbf{x}_{n-1})
\,&p_{\theta}(d\mathbf{x}_{n-1}| y_{0:n-1},\mathbf{x}_{n}) =\\&=
\int S_{\theta,n-1}(\mathbf{x}_{0:n-1})\, p_{\theta}(d\mathbf{x}_{0:n-2}| y_{0:n-2},\mathbf{x}_{n-1})\,p_{\theta}(d\mathbf{x}_{n-1}| y_{0:n-1},\mathbf{x}_{n}).
\end{align*}
Then, we have that,
\begin{align*}
p_{\theta}&(d\mathbf{x}_{0:n-2}| y_{0:n-2},\mathbf{x}_{n-1})\,p_{\theta}(d\mathbf{x}_{n-1}| y_{0:n-1},\mathbf{x}_{n}) \\
&\qquad \equiv 
p_{\theta}(d\mathbf{x}_{0:n-2}| \mathbf{x}_{n-1}, y_{0:n-1},\mathbf{x}_{n})\,p_{\theta}(d\mathbf{x}_{n-1}| y_{0:n-1},\mathbf{x}_{n})\\
&\qquad \equiv   p_{\theta}(d\mathbf{x}_{0:n-1}| y_{0:n-1},\mathbf{x}_{n}).
\end{align*}
Replacing the probability measure on the left side of the above equality with its equal on the right side, and using the latter in the integral above completes the proof for the first equation in the statement of the proposition.
The second equation follows from trivial use of Bayes rule.
\end{proof}

\section{Proof of Proposition 3} \label{proof3}
\begin{proof}
Part (i) follows via the same arguments as in \citet[Corollary 2]{olsson2017efficient}, based on \citet{azuma1967weighted} and \citet[Lemma 4]{douc2011sequential}.
For part (ii), one can proceed as follows.
Given $n\geq0$, we define the event  $A_{N}(1/j):=\{|\mathcal{S}_{\theta,n}-\hat{\mathcal{S}}_{\theta,n}^{N}|>\frac{1}{j}\}$,
where we added superscript $N$ to stress the dependency of the estimate
on the number of particles. Then,
\begin{align}
\mathrm{Prob}\,\big[\,\lim_{N\rightarrow\infty}\hat{\mathcal{S}}_{\theta,n}^{N}=\mathcal{S}_{\theta,n}\,\big]=
1-\mathrm{Prob}&\,\big[\,\cup_{j=1}^{\infty}\limsup_{N\rightarrow\infty}A_{N}(1/j)\,\big]\nonumber\\ 
&\geq 1-\sum_{j=1}^{\infty}\mathrm{Prob}\,\big[\,\limsup_{N\rightarrow\infty}A_{N}(1/j)\,\big].\label{eq:SS}
\end{align}
From (i), we get   $\mathrm{Prob}\,[\,A_{N}(1/j)\,]  \leq b_{n}e^{-c_{n}N\left(\frac{1}{j}\right)^{2}}$, so
$\sum_{N=1}^{\infty}\mathrm{Prob}\,[\,A_{N}(1/j)\,]<\infty$.
The Borel--Cantelli lemma gives $\mathrm{Prob}\,[\,\limsup_{N\rightarrow\infty}A_{N}(1/j)\,]=0$,
therefore the result follows from (\ref{eq:SS}).
\end{proof}

\end{document}